\documentclass[11pt]{article}

\usepackage[utf8]{inputenc}
\usepackage[T1]{fontenc}
\usepackage[UKenglish]{babel}

\usepackage{amsmath}
\usepackage{amssymb}
\usepackage{amsthm}
\usepackage[a4paper,top=25mm,left=23mm,%
right=23mm,bottom=25mm]{geometry}
\usepackage{graphicx}
\usepackage[dvipsnames,svgnames]{xcolor}
\usepackage{hyperref}
\usepackage{ebgaramond}
\usepackage{thm-restate}
\usepackage{authblk}
\usepackage{enumitem}
\usepackage{mathtools}
\usepackage{colortbl}
\usepackage{textgreek}
\usepackage{tikz}
\usepackage{xspace}

\author{Vincent Jugé}
\affil{LIGM, CNRS, Univ Gustave Eiffel, F77454 Marne-la-Vallée, France}

\title{Reduction ratio of the IS-algorithm:
worst and random cases}

\date{}
\usepackage{url}
\usepackage{bm}

\newcommand{\A}{\mathcal{A}}
\newcommand{\B}{\mathcal{B}}

\renewcommand{\O}{\mathcal{O}}
\renewcommand{\S}{\mathcal{S}}
\newcommand{\U}{\mathcal{U}}
\newcommand{\V}{\mathcal{V}}
\newcommand{\X}{\mathcal{X}}
\newcommand{\Y}{\mathcal{Y}}
\newcommand{\bE}{\bm{E}}
\newcommand{\bG}{\bm{G}}
\newcommand{\bM}{\bm{M}}
\newcommand{\bS}{\bm{\S}}
\newcommand{\bU}{\bm{\U}}
\newcommand{\bV}{\bm{\V}}
\newcommand{\bX}{\bm{\X}}
\newcommand{\bY}{\bm{Y}}

\newcommand{\bmm}{\bm{m}}
\newcommand{\bmu}{\bm{\mu}}
\newcommand{\bnu}{\bm{\nu}}
\newcommand{\bpi}{\bm{\pi}}
\newcommand{\bbE}{\mathbb{E}}
\newcommand{\bbP}{\mathbb{P}}
\newcommand{\bbR}{\mathbb{R}}
\newcommand{\bbU}{\mathbb{U}}

\newcommand{\is}{\mathsf{is}}
\newcommand{\eis}{\mathsf{eis}}
\newcommand{\src}{\mathsf{src}}
\newcommand{\lex}{<_{\text{lex}}}
\newcommand{\ovl}[1]{\mathring{#1}}
\newcommand{\eqdef}{\smash{\stackrel{\small\text{def}}{\,=\!=\,}}}
\newcommand{\medslash}{{%
\begin{tikzpicture}
\draw(0,0)--(0.065,0.12);
\end{tikzpicture}}}
\newcommand{\medbackslash}{{%
\begin{tikzpicture}
\draw(0.065,0.12)--(0.13,0);
\end{tikzpicture}}}

\newcommand{\epri}{\textsf{EPRI}\xspace}

\newtheorem{innerprop}{Proposition}
\newenvironment{myprop}[1]
  {\renewcommand\theinnerprop{#1}\innerprop}
  {\endinnerprop}
\newtheorem{innercor}{Corollary}
\newenvironment{mycor}[1]
  {\renewcommand\theinnercor{#1}\innercor}
  {\endinnercor}
\newtheorem{innerlem}{Lemma}
\newenvironment{mylem}[1]
  {\renewcommand\theinnerlem{#1}\innerlem}
  {\endinnerlem}

\newtheorem{theorem}{Theorem}
\newtheorem{lemma}[theorem]{Lemma}
\newtheorem{corollary}[theorem]{Corollary}
\newtheorem{proposition}[theorem]{Proposition}
\theoremstyle{definition}

\begin{document}


\maketitle

\begin{abstract}
We study the IS-algorithm,
a well-known linear-time algorithm for
computing the suffix array of a word.
This algorithm relies on transforming the input
word~$w$ into another word, called the \emph{reduced}
word of~$w$, that will be at least
twice shorter; then, the algorithm recursively
computes the suffix array of the reduced word.
In this article,
we study the \emph{reduction ratio} of the
IS-algorithm, i.e., the ratio between the lengths
of the input word and 
the word obtained after reducing~$k$
times the input word.
We investigate both worst cases, in which we
find precise results, and random cases,
where we prove some strong convergence phenomena.
Finally, we prove that, if the input word is a
randomly chosen word of length~$n$, we should not
expect much more than~$\log(\log(n))$ recursive
function calls.
\end{abstract}

\section{Introduction}
\label{sec:introdution}

The suffix array of a word
is the permutation of its suffixes that orders them
for the lexicographic order.
Suffix arrays were introduced in 1990 by
Manber and Meyers~\cite{manber1993suffix}
as a space-efficient alternative to suffix trees.
Like suffix trees, they have been used
since then in many applications~\cite{abouelhoda2004replacing,
crochemore2008simple,mozgovoy2005fast}:
data compression, pattern matching,
plagiarism detection,~\ldots

Suffix arrays were first constructed \emph{via}
the construction of suffix trees. Then, various
algorithms were proposed to construct suffix arrays
directly~\cite{karkkainen2017engineering,karkkainen2003simple,
kim2003linear, ko2005space}. A more comprehensive
list of approaches towards constructing suffix
trees can be found in~\cite{puglisi2007taxonomy}.
In 2010, a new algorithm, called the
\emph{IS-algorithm}, was proposed
for constructing suffix arrays~\cite{nong2010two}.
This algorithm, which is extremely efficient in
practice, is recursive:
except if the letters of its input word~$w$ are pairwise distinct,
in which case the suffix array of~$w$ is easy to compute directly,
the algorithm transforms~$w$ into
a shorter word~$w'$ and deduces
the suffix array of~$w$ from
the suffix array of~$w'$.

Thus, the question of knowing the
\emph{reduction ratio}~$|w'| / |w|$ between the
lengths of the words~$w'$ and~$w$, as well
as the number of recursive calls, is critical
to evaluating the efficiency of the algorithm.
More generally, denoting by~$\is^k(w)$ the word
obtained after~$k$ recursive calls (with~$\is^0(w) = w$),
we wish to evaluate the ratio~$|\is^k(w)| / |w|$ for all~$k$, as well as computing the number of recursive
calls that the algorithm will make, i.e.,
the maximal value of~$k$.

In this article, we focus on these two questions
in two different contexts.
In Section~\ref{sec:worst-case}, we consider
worst cases, and prove that
there exist arbitrarily long words~$w$ such that~$|\is^k(w)| \approx 2^{-k} |w|$ for all~$k \leqslant \log_2(|w|) - 3$.

Then, in Section~\ref{sec:markov-case},
we refine the work of~\cite{nicaud2015probabilistic}
and consider words whose letters are generated by
a Markov chain of order~1. In this context, and under mild conditions
about the Markov chain, we prove,
for each integer~$k \geqslant 0$, that the ratio~$|\is^k(w)| / |w|$ almost surely tends to a given
constant~$\gamma_k$ when~$|w| \to +\infty$.
Finally, in Section~\ref{sec:independent-case},
we study the constant~$\gamma_1$ (and, in some cases,~$\gamma_2$) when the letters
of~$w$ are identically and independently generated
and, in Section~\ref{sec:number-of-calls}, we
propose upper bounds on the number of recursive steps
on the IS-algorithm when the letters of~$w$ are
given by a finite Markov chain.

\section{Preliminaries}
\label{sec:preliminaries}

\subsection{Definitions and notations}
\label{sec:definitions}

Let~$\A$ be a non-empty alphabet,
endowed with a linear order~$\leqslant$.
For every integer~$n \geqslant 0$, we denote by~$\A^n$ the set of
words of length~$n$ over~$\A$, i.e., the set of
sequences of~$n$ letters in~$\A$.
We also denote by~$\A^\ast$ the set of all
finite words over~$\A$, i.e., the union~$\bigcup_{n \geqslant 0} \A^n$,
and by~$\varepsilon$ the empty word.

Let~$w$ be a finite word over~$\A$.
We denote by~$|w|$ the length of~$w$, and by~$w_0,w_1,\ldots,w_{|w|-1}$ the letters of~$w$.
We may abusively denote by~$w_{-k}$ the letter~$w_{|w|-k}$, i.e., the~$k$\textsuperscript{th}
rightmost letter of~$w$.
For all integers~$i$ and~$j$ such that~$0 \leqslant i \leqslant j \leqslant |w|-1$, we also
denote by~$w_{i \cdots j}$ the word~$w_i w_{i+1} \cdots w_j$.
Every such word is called a \emph{factor} of~$w$.
If~$j = |w|-1$,
this word is a suffix of~$w$, and we
also denote it by~$w_{i \cdots}$.
Finally, given two words~$u$ and~$v$, we denote
by~${u \cdot v}$ their concatenation, i.e., the word~$u_0 u_1 \cdots u_{|u|-1} v_0 v_1 \cdots v_{|v|-1}$.

The \emph{suffix array}~\cite{manber1993suffix}
of a word~$w \in \A^\ast$
is the unique permutation~$\sigma$
of~$\{0,1,\ldots,|w|-1\}$ such that~$w_{\sigma(0)\cdots} \lex w_{\sigma(1)\cdots} \lex
\ldots \lex w_{\sigma(|w|-1)\cdots}$, where~$\lex$ denotes the lexicographic ordering.
The IS-algorithm~\cite{nong2010two} aims at
computing the suffix array of its input word~$w$
in time linear in~$|w|$,
when the alphabet~$\A$ is either a given finite set
or a subset of~$\{0,1,\ldots,|w|-1\}$.

\subsection{Unimodal factors and one-step reduction}
\label{sec:unimodal}

Let~$w$ be a finite word over~$\A$,
and let~$\$$ be a fictitious letter, called the
\emph{sentinel}, that is defined to be smaller
than all letters in~$\A$. Below, we simply
denote by~$\A_\$$ the set~$\A \cup \{\$\}$.

An integer~$i \leqslant |w|-1$
is said to be \emph{$w$-non-decreasing}
if there exists an integer~$j$ such that~$i+1 \leqslant j \leqslant |w|-1$ and~$w_i = w_{i+1} = \ldots = w_{j-1} < w_j$.
If, in addition,~$i \geqslant 1$ and~$w_{i-1} > w_i$, we say that~$i$ is
\emph{$w$-locally minimal}.

Then, let~$i_0 < i_1 < \ldots < i_{k-1}$ be the~$w$-locally minimal integers (with~$k \geqslant 0$).
We also set~$i_k = |w|$, and we abusively
set~$w_{|w|} = \$$. This amounts
to replacing~$w$ by the word~$w \cdot \$$,
whose suffix array is the same as the one of~$w$,
except that we appended the letter~$\$$ to every
suffix and that~$\$$ is now the least non-empty
suffix of~$w \cdot \$$.

We define the \emph{unimodal
factors} of~$w$, also called
\emph{LMS factors}~\cite{nicaud2015probabilistic,%
nong2010two},
as the~$k$ words~$w_{i_0 \cdots i_1}$, $w_{i_1 \cdots i_2}$,~\ldots,
$w_{i_{k-1} \cdots i_k}$, which belong to~$\A^+ \cdot (\varepsilon + \$)$.
We call these factors \emph{unimodal} because
each sequence~$w_{i_\ell},w_{i_\ell+1},
\ldots,w_{i_{\ell+1}}$ consists of a non-decreasing prefix
followed by a non-increasing suffix,
and we denote by~$\eis(w)$ --- for \emph{expanded
IS-reduction} of~$w$ --- the word over the infinite
alphabet~$\A^+ \cdot (\varepsilon + \$)$
whose letters are the unimodal factors of~$w$.

For instance, if~$w$ is the word \texttt{COMBINATORIAL}
over the latin alphabet~$\A$,
its unimodal factors are \texttt{BINA}, \texttt{ATO},
\texttt{ORIA} and \texttt{AL$\$$},
and thus~$\eis(w)$ is the four-letter word~$\texttt{BINA} \cdot \texttt{ATO} \cdot \texttt{ORIA} \cdot \texttt{AL\$}$
over the alphabet~$\A^+ \cdot (\varepsilon + \$)$.

In subsequent sections, we may extend to infinite words~$w$
(to which we append the letter~$\$$ if~$w$ is left-infinite,
but not if~$w$ is right-infinite)
the notions of~$w$-locally minimal integer, of unimodal factor,
and of expanded IS-reduction.

The IS-algorithm roughly works as follows:
\begin{enumerate}[itemsep=0pt]
\item\label{step:1}
compute~$w$-locally minimal integers and
the associated unimodal factors,
which form the letters of~$\eis(w)$;
\item\label{step:2}
sort these factors;
\item\label{step:3}
if~$w$ has~$\ell$ distinct unimodal factors,
identify each factor with an integer~$i \in \{0,1,\ldots,\ell-1\}$: factors~$f$ and~$f'$ such that~$f \lex f'$ are identified
with integers~$i$ and~$i'$ such that~$i < i'$;
\item\label{step:4}
identify the word~$\eis(w)$ with
a word~$\is(w)$ over the alphabet~${\{0,1,\ldots,\ell-1\}}$;
\item\label{step:5}
compute the suffix array of~$\is(w)$, either
directly (if the letters of~$\is(w)$ are pairwise distinct)
or recursively (if at least two letters of~$\is(w)$ coincide
with each other);
\item\label{step:6}
based on that array, sort all suffixes of~$w$.
\end{enumerate}

As mentioned by its authors~\cite{nong2010two},
steps \ref{step:1}, \ref{step:3} and \ref{step:4} of the algorithm
can clearly be performed in time~$\O(|w|)$. If~$\A$ is a given finite set,
or a subset of~$\{0,1,\ldots,|w|-1\}$,
bucket sorts allow sorting in linear time unimodal words
whose rightmost letters are already sorted,
thereby performing
steps \ref{step:2} and \ref{step:6}
in time~$\O(|w|)$. Finally,
no two consecutive integers~$i \leqslant |w|-1$
are~$w$-locally minimal, and therefore~$|\is(w)| \leqslant |w|/2$, thereby proving that
the IS-algorithm works in time~$\O(|w|)$.

Thus, a natural question would be that of evaluating
the constant hidden in this~$\O(|w|)$ running time.
To that end, we could focus closely on how each of the
steps \ref{step:1} to \ref{step:4} and \ref{step:6}
is performed. However,
several variants might be considered for performing
each of these steps. Consequently, we focus
on the step~\ref{step:5} and study
the behaviour of the ratio~$|\is(w)| / |w|$ or, more generally,~$|\is^k(w)| / |w|$.

\subsection{Markov chains and ergodicity}
\label{sec:probabilities}

In Sections~\ref{sec:markov-case}
to~\ref{sec:number-of-calls}, we consider
\emph{random} words, whose letters result from
a probabilistic process, and are random variables
that form a (homogeneous)
\emph{Markov chain} of order 1. Below, we focus exclusively
on such Markov chains, and thus abandon
the words ``homogeneous'' and ``of order 1''.

Let~$\S$ be a countable set,
let~$\mu : \S \mapsto \bbR$ be a probability
distribution, and let~$M : \S \times \S \mapsto \bbR$ be a
function such that~$\sum_{t \in \S} M(s,t) = 1$
for all~$s \in \S$.
A homogeneous Markov chain of order 1
with \emph{set of states}~$\S$, \emph{initial distribution}~$\mu$
and \emph{transition matrix}~$M$ is
a sequence of random variables~$(X_n)_{n \geqslant 0}$
with values in~$\S$ such that~$\bbP(X_0 = x) = \mu(x)$ for all~$x \in \S$ and
such that, for
every integer~$n \geqslant 1$ and every tuple~$(x_0,x_1,\ldots,x_n) \in \S^{n+1}$, we have
\[\bbP(X_n = x_n \mid X_0 = x_0,
X_1 = x_1,\ldots,X_{n-1} = x_{n-1}) =
M(x_{n-1},x_n)\]
whenever~$\bbP(X_0 = x_0,
X_1 = x_1,\ldots,X_{n-1} = x_{n-1}) > 0$.
Below, we identify the Markov chain with the pair~$(M,\mu)$, or with the transition matrix~$M$
in contexts where the initial distribution is
irrelevant and might need to be changed.
We also abusively say that~$(X_n)_{n \geqslant 0}$
is a \emph{trajectory} of the Markov chain~$(M,\mu)$ or, alternatively, is \emph{generated} by~$(M,\mu)$.

The \emph{underlying graph} of~$(M,\mu)$
is the weighted graph~$G = (V,E,\pi)$ with
vertex set~${V = \S}$, edge set
\[E = \{(s,t) \in \S \times \S \colon
M(s,t) > 0\},\]
and whose weight function~$\pi : E \mapsto \bbR$ is defined by~$\pi(s,t) = M(s,t)$.
We say that~$(M,\mu)$ is \emph{irreducible} if~$G$
is strongly connected,
and \emph{aperiodic} when the lengths of its cycles
have no common divisor~$d \geqslant 2$.

These notions are connected to the
\emph{ergodicity} of a Markov chain,
which can be defined as follows.
Given a probability distribution~$\nu$ on~$\S$,
we denote by~$M \nu$ the probability distribution
defined by
\[(M \nu)(x) =
\sum_{y \in \S} M(y,x) \nu(y).\]
The~$L^1$ distance between two distributions~$\nu$ and~$\theta$ is defined as the real number~$\|\nu-\theta\|_1 = \sum_{x \in \S} |\nu(x) - \theta(x)|$.
The Markov chain~$M$ is said to be
\emph{ergodic}
if there exists a positive probability
distribution~$\nu$ on~$\S$
(i.e., a probability
distribution such that~$\nu(x) > 0$ for all~$x \in \S)$ such that~$\lim_{k \to +\infty} \|\nu - M^k \theta\|_1 = 0$
for all probability distributions~$\theta$ on~$\S$.

Such a distribution~$\nu$
must be the unique \emph{stationary
distribution} of the Markov chain~$M$,
i.e., the unique probability distribution such that~$\nu = M \nu$.
Conversely, when~$M$ is irreducible and has a
stationary distribution that is positive on~$\S$,
we say that~$M$ is \emph{irreducible and positive recurrent}.
This latter assumption relieves us from the need of
aperiodicity, and yet retains some desirable properties
of ergodic Markov chains.

A typical example of an ergodic Markov chain
arises if~$\bbP({X_n = t \mid} {X_{n-1} = s}) = {\nu(t)}$
for all~$s$ and~$t$ in~$\S$, i.e.,
if~$M \theta = \nu$ for all probability distributions~$\theta$ on~$\S$.
In that case, the random variables~$(X_n)_{n \geqslant 0}$ are said to be
\emph{independent and identically distributed}.

We refer the reader to~\cite{levin2017markov,porod21}
for a comprehensive
review about Markov chains and their properties,
from which we present three crucial results below.

\begin{proposition}[Corollary~1.18 and
Theorem~21.14 of~\cite{levin2017markov}]
\label{pro:1.18}
Every ergodic Markov chain is irreducible
and aperiodic. Conversely,
every irreducible and aperiodic Markov chain
is ergodic, provided that its state set is finite
or that it has a positive stationary distribution.
\end{proposition}

We are particularly interested in
Theorem~4.16 of~\cite{levin2017markov}, on which
we will base Section~\ref{sec:markov-case}.
However, we will not necessarily handle ergodic
Markov chains, and therefore we shall relax
the notion of ergodicity to a less stringent,
\emph{ad hoc} notion that we call
\emph{almost surely eventually positive recurrent and irreducible}
(or \epri) Markov chains.

A Markov chain~$M$ with underlying graph~$G = (\S,E,\pi)$,
is said to be \epri if there exists
a set~$\X \subseteq \S$ of states,
called the \emph{terminal component}
of~$M$, such that
(i)~$\X$ is a strongly connected component of~$G$;
(ii)~$M$~has~stationary distribution~$\nu$,
i.e., a probability distribution~$\nu$ such that~$M \nu = \nu$, that is positive on~$\X$ and zero
on~$\S \setminus \X$; and
(iii)~for every initial distribution~$\mu$,
the sequence generated by~$(M,\mu)$ almost
surely contains a vertex~${x \in \X}$.

Note that, since~$\nu$ is positive on~$\X$ and zero
elsewhere, no edge of~$G$ can leave~$\X$, i.e.,
the set~$E$ contains no edge~$(x,y)$ such that~$x \in \X$
and~$y \notin \X$.

In this notion, we completely abandon
any requirement to be acyclic, which prevents
the~$L^1$ convergence that characterises ergodicity.
However, when focusing on \emph{average, long-term}
behaviours of a Markov chain, such as the frequency of
occurrence of a given vertex of sequence of consecutive vertices,
whether the Markov chain is cyclic or acyclic is irrelevant.
Thus, we may just focus on
irreducible, positive recurrent Markov chains.
Moreover, in \epri Markov chains,
the path followed before entering the terminal component
quickly vanishes.
Consequently, the following result, which is usually stated
for irreducible, positive recurrent Markov chains only,
can be generalised to all
\epri Markov chains whose state space is
either finite or countably infinite.

\begin{theorem}[Theorem~4.16 of~\cite{levin2017markov},
Theorem~2.1.1 of~\cite{porod21}]
\label{thm:4.16}
Let~$(M,\mu) = (X_n)_{n \geqslant 0}$ be an \epri
Markov chain
with set of states~$\S$ and stationary distribution~$\nu$.
Let~$\ell$ be a positive integer,~$f : \S^\ell \mapsto \bbR$
be a bounded function, and
\[\bbE_\nu[f] =
\sum_{x_1,x_2,\ldots,x_\ell \in \S}
\nu(x_1) M(x_1,x_2) M(x_2,x_3) \cdots M(x_{\ell-1},x_\ell) f(x_1,x_2,\ldots,x_\ell).\]
We have
\[\bbP\left[\frac{1}{n}\sum_{k=0}^{n-1}f(X_k,X_{k+1},\ldots,X_{k+\ell-1}) \xrightarrow{n \to +\infty} \bbE_\nu[f]\right] = 1.\]
\end{theorem}

\begin{proof}
It is well-known~\cite{porod21} that Theorem~\ref{thm:4.16}
holds when~$M$ is irreducible and positive recurrent, i.e.,
when its state space~$\S$
coincides with its terminal component~$\X$.

In the general case, trajectories of the Markov chain
almost surely meet~$\X$ after a finite number
of steps, say~$p$, that depends of the trajectory.
Once it meets~$\X$,
the trajectory starts behaving like an irreducible, positive
recurrent Markov chain with state space~$\X$, and therefore
\[\frac{1}{n-p}\sum_{k=p}^{n-1}
f(X_k,X_{k+1},\ldots,X_{k+\ell-1})\]
converges almost surely (as~$n \to +\infty$) to~$\bbE_\nu[f]$.
Theorem~\ref{thm:4.16} follows.
\end{proof}

Finally, a crucial well-known property of irreducible, 
positive recurrent Markov chains
whose initial distribution coincides with their
stationary distribution is that they can be
\emph{reversed}.

\begin{theorem}[Proposition~1.22 of~\cite{levin2017markov}]
\label{thm:1.22}
Let~$(X_n)_{n \geqslant 0}$ be  an irreducible, 
positive recurrent  Markov chain
with set of states~$\S$,
transition matrix~$M$,
and whose initial distribution coincides
with the stationary distribution~$\nu$ of~$M$.
For all integers~$\ell \geqslant 0$, the sequence~$(X_{\ell-n})_{0 \leqslant n \leqslant \ell}$
contains the first~$\ell+1$ elements of an irreducible, 
positive recurrent Markov chain, called the
\emph{reverse} Markov chain of~$(M,\nu)$,
with initial distribution~$\nu$ and
whose transition matrix~$\hat{M}$ is defined by
\[\hat{M}(x,y) = \frac{\nu(y)}{\nu(x)} M(y,x).\]
\end{theorem}

More generally, if~$M$ is \epri, and provided that its
initial distribution is~$\nu$, it already starts inside of
its terminal component~$\X$,
which it cannot leave. Thus, up to
deleting those states of~$M$ that do not belong to~$\X$,
the Markov chain~$M$ becomes irreducible and positive
recurrent, and Theorem~\ref{thm:1.22} applies, with the
following caveat: the state space of its reverse Markov
chain is restricted to~$\X$,
and needs not be extended to states outside of~$\X$.

\section{Deterministic worst case}
\label{sec:worst-case}

By construction, no two consecutive integers~$i \leqslant |w|-1$ are~$w$-locally minimal,
and all \mbox{$w$-locally} minimal integers belong to the
set~$\{1,2,\ldots,|w|-2\}$.
Hence, at most~$(|w|-1)/2$ integers are~$w$-locally minimal. This means that~$|\is(w)|+1 \leqslant (|w|+1) / 2$ and, more generally, that~$|\is^k(w)|+1 \leqslant 2^{-k} (|w|+1)$
for every integer~$k \geqslant 0$
and every word~$w \in \A^\ast$ such that~$\is^k(w)$ exists.
A genuine question is then:
can we do better?
The answer, which was known to be negative~\cite{bingmann2016inducing} when we allow
alphabets~$\A$ with size~$\log_2(|w|)$,
remains negative for every fixed size~$|A| \geqslant 2$.

\begin{theorem}
\label{thm:1}
Let~$\A$ be an alphabet of cardinality at least~$4$.
For every integer~$n \geqslant 3$, there exists
a word~$w \in \A^{2^n-1}$ on which
the IS-algorithm performs~$n-2$ recursive calls,
and
\[|\is^k(w)|+1 = 2^{-k} (|w|+1)\]
for all~$k \in \{0,1,\ldots,n-2\}$.
\end{theorem}

\begin{proof}
Without loss of generality, we assume that~$\A = \{0,1,2,4\}$.
Let also~$\B = \{0,1,2,3,4\}$. Then, let~$\varphi \colon \B^\ast \mapsto \B^\ast$ and~$\psi \colon \B^\ast \mapsto \A^\ast$ 
be morphisms of monoids, uniquely defined by their
values on~$\B$:~$\varphi(0) = 02$,~$\varphi(1) = 04$,~$\varphi(2) = 12$,~$\varphi(3) = 13$ and~$\varphi(4) = 14$;~$\psi(a) = a$ for all~$a \in \A$, and~$\psi(3) = 4$.
We prove below that the word~$\psi(\varphi^n(3)_{1 \cdots})$
satisfies the requirements of Theorem~\ref{thm:1}.

We say that a word~$w = w_0 w_1 \cdots w_k \in \B^\ast$ is
\emph{balanced} if (1) its length~$|w| = k+1$
is even,
(2) its rightmost letter~$w_k = 3$,
(3) its suffix~$w_{1\cdots}$ contains each of the letters~$0,1,2,3,4$,
and (4) for all~$i \leqslant k-1$, we have~$w_i \in \{0,1\}$ if~$i$ is even and~$w_i \in \{2,4\}$ if~$i$ is odd.
The eight-letter word~$\varphi^3(3) = 02140413$
is balanced, and~$\varphi$ maps each balanced
word to a balanced word.

Provided that~$w$ is balanced, the~$\varphi(w)_{1 \cdots}$-minimal integers are~$1,3,5,\ldots,2k-1$, and the associated unimodal factors are~$\varphi(w_1) \cdot \varphi(w_2)_0,
\varphi(w_2) \cdot \varphi(w_3)_0,\ldots,
\varphi(w_{k-1}) \cdot \varphi(w_k)_0,
{\varphi(w_k) \cdot \$}$.
Since~$\varphi(0)_1 = \varphi(1)_1 = 0$ and~$\varphi(2)_1 = \varphi(3)_1 = \varphi(4)_1 = 1$,
this means that the unimodal factors of~$\varphi(w)$ are~$\theta(w_1),\theta(w_2),\ldots,\theta(w_k)$,
where we set~$\theta(0) = 021$,~$\theta(1) = 041$,~$\theta(2) = 120$,~$\theta(3) = 13\$$ and~$\theta(4) = 140$.
The function~$\theta$ is increasing, and thus,~$\is(\varphi(w)_{1 \cdots}) =
w_{1 \cdots}$.

Moreover, if~$w$ is balanced,
and since the rightmost letter of~$\varphi(w)$
is its only occurrence of the letter~$3$,
the words~$\varphi(w)_{1\cdots}$ and~$\psi(\varphi(w)_{1\cdots})$ have the same unimodal factors, except that their last factors
are~$13\$$ and~$14\$$, respectively. Hence,~$\is(\psi(\varphi(w)_{1 \cdots})) =
\is(\varphi(w)_{1 \cdots}) =
w_{1 \cdots}$.
Thus, the map~$\is$ successively sends~$\psi(\varphi^n(3)_{1 \cdots})$ to~$\varphi^{n-1}(3)_{1 \cdots},
\varphi^{n-2}(3)_{1 \cdots},
\ldots,
\varphi^3(3)_{1 \cdots}$,
and observing that~$\is(\varphi^3(3)_{1 \cdots}) = 201$
completes the proof.
\end{proof}

Although the conclusions of Theorem~\ref{thm:1} are not
valid for alphabets of cardinality~$2$ or~$3$,
it is still possible to find variants of this
worst case. In these variants, the first step of
the IS-algorithm is more efficient,
with respective reduction ratios of~$3$ and~$5/2$,
but every word considered after that first step
belongs to an alphabet of cardinality~$4$,
which explains why the reduction ratios we compute
have similar orders of magnitude.

\begin{corollary}
\label{cor:2}
Let~$\A$ be an alphabet of cardinality~$2$.
For every integer~$n \geqslant 3$, there exists
a word~$w \in \A^{3 \times 2^n-2}$ on which
the IS-algorithm performs~$n-1$ recursive calls,
and
\[|\is^k(w)|+1 = 2^{1-k} (|w|+2) / 3\]
for all~$k \in \{1,2,\ldots,n-1\}$.
\end{corollary}

\begin{proof}
Let us assume that~$\A = \{0,1\}$,
and let~$\B$ and~$\varphi$ be the
alphabet and the morphism defined in
the proof of Theorem~\ref{thm:1}.
An immediate induction on~$\ell$ shows that,
for all~$\ell \geqslant 3$,
the word~$\varphi^\ell(3)$
starts with the letter~$0$, ends with the letter~$3$, and contains~$2^{\ell-2}$ letters~$0$,~$2^{\ell-2}$ letters~$1$,~$2^{\ell-2}-1$ letters~$2$, one letter~$3$
(the rightmost one) and~$2^{\ell-2}$ letters~$4$.

Then, we consider a new morphism~$\psi_2 \colon
\B^\ast \mapsto \A^\ast$, such that~$\psi_2(0) = 0001$,~$\psi_2(1) = 001$,~$\psi_2(2) = 01$, and~$\psi_2(3) = \psi_2(4) = 011$.
Like in the proof of Theorem~\ref{thm:1}, we prove
that
\[\is(1 \cdot \psi_2(w_{1 \cdots})) =
\is(\varphi(w)_{1 \cdots}) =
w_{1 \cdots}\] when~$w$ is balanced,
and having counted occurrences of each letter
in~$\varphi^n(3)$ allows us
to conclude that the word~$1 \cdot \psi_2(\varphi^n(3)_{1 \cdots})$
satisfies the requirements of
Corollary~\ref{cor:2}.
\end{proof}

\begin{corollary}
\label{cor:3}
Let~$\A$ be an alphabet of cardinality~$3$.
For every integer~$n \geqslant 3$, there exists
a word~$w \in \A^{5 \times 2^n-3}$ on which
the IS-algorithm performs~$n$ recursive calls,
and
\[|\is^k(w)|+1 = 2^{2-k} (|w|+3) / 5\]
for all~$k \in \{1,2,\ldots,n\}$.
\end{corollary}

\begin{proof}
The proof is the same as that of
Corollary~\ref{cor:2},
except that we have now~${\A = \{0,1,2\}}$
and that, instead of the morphism~$\psi_2$, we use a new morphism~$\psi_3 \colon
\B^\ast \mapsto \A^\ast$, such that~$\psi_3(0) = 001$,~$\psi_3(1) = 01$,~$\psi_3(2) = 012$ and~$\psi_3(3) = \psi_3(4) = 02$.
Indeed, we also
have
\[\is(1 \cdot \psi_3(w_{1 \cdots})) =
\is(\varphi(w)_{1 \cdots}) =
w_{1 \cdots}\] when~$w$ is balanced,
from which we conclude that the word~$1 \cdot \psi_3(\varphi^n(3)_{1 \cdots})$
satisfies the requirements of
Corollary~\ref{cor:3}.
\end{proof}

\section{Words generated by an ergodic Markov chain}
\label{sec:markov-case}

Let~$\A$ be a finite or countably infinite set.
Below, we study the typical behaviour of the IS-algorithm
on a word~$w \in \A^n$ whose letters are the first~$n$ elements of
an \epri  Markov chain~$(M,\mu)$ with set of states~$\A$.
We prove below the following result,
which is the main (and technically most demanding)
result presented in this paper.

\begin{theorem}
\label{thm:6}
Provided that~$w$ is generated by an \epri Markov chain,
and for all integers~$k \geqslant 0$,
there exist a constant~$\gamma_k$ and a sequence~$(\varepsilon_n)_{n \geqslant 0}$
that tends to~$0$ such that
\[\bbP\left[\left|\frac{|\is^k(w)|}{|w|}
- \gamma_k\right| \geqslant \varepsilon_{|w|}\right]
\leqslant \varepsilon_{|w|}.\]
\end{theorem}

A particular case of interest arises when~$w$ is a word
over a finite alphabet generated by an ergodic Markov chain.
However, even in that restricted case,
studying the words~$\is^k(w)$ for~$k \geqslant 1$
will require us to consider words over infinite alphabets,
which might be generated by Markov chains no longer
ergodic, but only \epri.
That is why, facing the need to treat such a generalised
setting, we chose to include it from the start in our study.

In addition, all finite-state Markov chains can be decomposed
as a ``sum'' of \epri Markov chains. Indeed, if the underlying
graph of such a Markov chain~$(M,\mu)$ has~$k$ terminal
strongly connected components, the Markov chain will almost
surely reach one of these components. Thus, in order to study the
Markov chain~$(M,\mu)$, we may consider, one by one, its~$k$ terminal components; for each such component~$K$,
compute the probability that~$(M,\mu)$ eventually reaches~$K$;
finally, simulate the behaviour of~$(M,\mu)$ by first selecting
at random which terminal component~$K$ it will reach, and then
assuming that~$(M,\mu)$ must reach that component, thereby
transforming~$(M,\mu)$ into an \epri Markov chain.
This allows us to obtain the following variant of
Theorem~\ref{thm:6}.

\begin{theorem}
\label{thm:6a}
Let~$w$ be a word whose letters are generated by a finite-state
Markov chain. There exist a constant~$\kappa$ and a probability
law~$X$ over the set~$\{1,2,\ldots,\kappa\}$ with the
following property:
For all integers~$k \geqslant 0$,
there exist constants~$\gamma_{1,k},\gamma_{2,k},\ldots,
\gamma_{\kappa,k}$ and a sequence~$(\varepsilon_n)_{n \geqslant 0}$
that tends to~$0$ such that, for all~$i \leqslant \kappa$,
\[\left|\bbP\left[\left|\frac{|\is^k(w)|}{|w|}
- \gamma_{i,k}\right| \leqslant \varepsilon_{|w|}\right]
- \bbP[X = i]\right| \leqslant \varepsilon_{|w|}.\]
\end{theorem}
 
\subsection{Generating letters from right to left}
\label{sec:right-left}

In~\cite{nicaud2015probabilistic}, the letters
of~$w$ are generated from right to left, i.e.,
the letter~$w_{n-k}$ is the~$k$\textsuperscript{th} element
of the Markov chain.
Here, we mainly focus on this case too.
Generating the letters of~$w$ from right to left
makes things easier
because, although being~$w$-non-decreasing
is \emph{not} a local property,
it enjoys the following local,
recursive characterization:
an integer~$i$ is~$w$-non-decreasing if and only if~$i \leqslant |w|-2$ and
either (a)~$w_i < w_{i+1}$,
or (b)~$w_i = w_{i+1}$ and~$i+1$ is~$w$-non-decreasing.

Below, we wish to study the sequence~$w, \is(w), \is^2(w), \ldots$
and in particular the lengths of these words.
In fact, it will be easier to study the sequence~$w, \eis(w), \eis^2(w), \ldots$
These two sequences differ from each other
because they do not use the same alphabets. Yet, for
all~$k \geqslant 0$, the words~$\is^k(w)$ and~$\eis^k(w)$ are ``isomorphic'' to each other:
they have the same length, and
there exists an increasing
mapping~$\varphi$
from the letters of~$\eis^k(w)$ to those
of~$\is^k(w)$, such that~$\varphi(\eis^k(w)_i) = \is^k(w)_i$ for all~$i < |\eis^k(w)|$.

Following~\cite{nicaud2015probabilistic,nong2010two},
we transform the Markov chain~$(M,\mu)$
into another Markov chain~$(\overline{M},\overline{\mu})$ that starts
with the letter~$\$$ and,
in addition to telling which letter we produce,
also tells whether the corresponding index
is~$w$-non-decreasing: instead of producing
letters~$a \in \A_\$$, this new Markov chain
shall produce pairs~$(a,\uparrow)$ or~$(a,\downarrow)$,
depending on whether the current position
is~$w$-non-decreasing or not:
we produce a pair~$(a,\uparrow)$ if the former case,
and~$(a,\downarrow)$ in the latter case.
Formally, the Markov chain~$(\overline{M},\overline{\mu})$ is defined as follows.
Its states form the set~$\overline{\S} = \A_\$ \times \{\uparrow,\downarrow\}$.
Its initial distribution is defined by ~$\overline{\mu}(\$,\uparrow) = 1$, and~$\overline{\mu}(s) = 0$ whenever~$s \neq (\$,\uparrow)$.
Its transition matrix is then defined by
\[
\begin{cases}
\overline{M}\big((\$,\updownarrow),(y,\downarrow)\big)
= \mu(y) & \text{if } y \in \A; \\
\overline{M}\big((x,\updownarrow),
(y,\downarrow)\big) = M(x,y)
& \text{if } (x,y) \in \A^2 \text{ and } x < y; \\
\overline{M}\big((x,\updownarrow),
(y,\uparrow)\big) = M(x,y)
& \text{if } (x,y) \in \A^2 \text{ and } x > y; \\
\overline{M}\big((x,\updownarrow),
(y,\Updownarrow)\big) = M(x,y)
& \text{if } (x,y) \in \A^2 \text{, } x = y \text{ and }
\updownarrow=\Updownarrow; \\
\overline{M}\big((x,\updownarrow),(y,\Updownarrow)) = 0
& \text{otherwise}.
\end{cases}\]

\begin{proposition}
\label{pro:7}
Let~$(M,\mu)$ be an \epri Markov chain
whose terminal component
has size at least two.
The Markov chain~$(\overline{M},\overline{\mu})$
defined above is \epri.
\end{proposition}

\begin{proof}
Let~$G = (\A,E,\pi)$
be the underlying graph of the Markov chain~$(M,\mu)$, let~$\X$ be its terminal component,
and let~$\nu$ be its stationary distribution.
In addition, for all~$x \in \A$, let ~$x^\uparrow = \{y \in \X \colon x < y \text{ and }
(y,x) \in E\}$ and~$x^\downarrow = \{y \in \X \colon x > y \text{ and }
(y,x) \in E\}$.

Since~$M(x,x)< 1$ for all~$x \in \A$,
the distribution ~$\overline{\nu}$ on~$\overline{\S}$ defined by~$\overline{\nu}(\$,\updownarrow) = 0$ and by
\[
\overline{\nu}(x,\updownarrow) = 
\frac{1}{1 - M(x,x)} \sum_{y \in x^\updownarrow}
M(y,x) \nu(y)\]
for all~$(x,\updownarrow) \in \A \times \{\uparrow,\downarrow\}$
is a probability distribution, because
\[
\overline{\nu}(x,\uparrow) + \overline{\nu}(x,\downarrow)
= \frac{1}{1-M(x,x)} \sum_{y \colon \! x \neq y}
M(y,x) \nu(y)
= \frac{M \nu (x) - M(x,x) \nu(x)}{1-M(x,x)}
= \nu(x) \tag{1}\label{eq:1}\]
for all~$x \in \A$. We further deduce from~\eqref{eq:1}
that
\begin{align*}
\overline{M} \overline{\nu}(x,\updownarrow) -
M(x,x) \overline{\nu}(x,\updownarrow)
= \sum_{y \in x^{\updownarrow}} M(y,x)
\big(\overline{\nu}(y,\uparrow) +
\overline{\nu}(y,\downarrow)\big)
& = \sum_{y \in x^{\updownarrow}} M(y,x) \nu(y) \\
& = (1 - M(x,x)) \overline{\nu}(x,\updownarrow),\end{align*}
i.e., that~$\overline{M} \overline{\nu}(x,\updownarrow) =
\overline{\nu}(x,\updownarrow)$, for all~$(x,\updownarrow) \in \A \times \{\uparrow,\downarrow\}$.
This means that~$\overline{\nu}$ is a stationary distribution of~$(\overline{M},\overline{\mu})$.

This probability distribution is positive
on the set
\[\overline{\X} \eqdef
\{(x,\uparrow) \colon x \in \X,
x^\uparrow \neq \emptyset\} \cup
\{(x,\downarrow) \colon x \in \X, x^\downarrow \neq
\emptyset\}\]
and is zero
outside of~$\overline{\X}$.
Since~$\overline{\nu}$ is non-zero, it follows that~$\overline{\X}$ is non-empty.

Then, let~$\overline{G}$ be the
underlying graph of~$(\overline{M},\overline{\mu})$.
We shall prove that~$\overline{\X}$ satisfies the
requirements~(i) and~(iii) of \epri
Markov chains.
Hence, consider some state~$(x,\uparrow)$ in~$\overline{\X}$,
and let~$y$ be a state in~$x^\uparrow$.
For every state~$(z,\updownarrow)$ in~$\overline{\X}$,
the graph~$G$ contains a finite path from~$z$ to~$x$
whose second-to-last vertex is~$y$,
and thus~$\overline{G}$ contains a finite path
from~$(z,\updownarrow)$ to~${(x,\uparrow)}$.
Similarly, every state~$(x,\downarrow)$ in~$\overline{\X}$
is accessible from every state~$(z,\updownarrow)$
in~$\overline{\X}$, and thus~$\overline{\X}$
satisfies the requirement~(i).

Finally, consider some trajectory~$(\overline{X}_n)_{n \geqslant 0}$
of~$(\overline{M},\overline{\mu})$.
Deleting its first vertex and removing the second
component of each vertex transforms~$(\overline{X}_n)_{n \geqslant 0}$ into
a trajectory~$(X_n)_{n \geqslant 1}$ of the Markov chain~$M$,
which almost surely contains a vertex~$x \in \X$ and then
almost surely meets a vertex distinct from~$x$; let~$y$ be the first such vertex.
The trajectory~$(\overline{X}_n)_{n \geqslant 0}$
contains the vertex~$(y,\uparrow)$ if~$y < x$, or~$(y,\downarrow)$
if~$y > x$, and in both cases that vertex belongs
to~$\overline{\X}$. This shows that~$\overline{\X}$ satisfies the requirement~(iii).
\end{proof}

Using Theorem~\ref{thm:4.16} for the function~$f : \overline{\S} \times \overline{\S} \mapsto \bbR$
defined by 
\[\begin{cases}
\displaystyle
f\big((x,\uparrow),(y,\downarrow)\big) = 1 &
\text{for all } x, y \in \A; \\
f(u,v) = 0 &
\text{in all other cases}
\end{cases}\]
already allows us to prove a special case
of Theorem~\ref{thm:6} for~$k = 1$, which was already
proven in~\cite{nicaud2015probabilistic} in the case~$\A$
is finite and~$(M,\mu)$ is ergodic.

However, if the terminal component of~$M$
contains only one state~$z$,
the Markov chain~${(\overline{M},\overline{\mu})}$
is no longer \epri, since its graph contains two self-loops
around~$(z,\uparrow)$ and~${(z,\downarrow)}$, each one with
weight~$1$. We overcome this difficulty by merging
the two states~$(z,\uparrow)$ and~${(z,\downarrow)}$ into one
single state~$z$, thereby recovering an \epri Markov chain,
and we modify the function~$f$, redefining it by
\[\begin{cases}
\displaystyle
f\big((x,\uparrow),(y,\downarrow)\big) = 1 &
\text{for all } x, y \in \A \setminus \{z\}; \\
f\big((x,\uparrow),z\big) = 1 &
\text{for all } x \in z^\downarrow; \\
f(u,v) = 0 &
\text{in all other cases.}
\end{cases}\]

Tackling this special case allows us to derive the following
result, whose validity does not depend on the
size of the terminal component of~$M$.

\begin{corollary}
\label{cor:8}
If the letters of~$w$ are generated from right to left
by an \epri Markov chain,
there exists a constant~$\gamma_1$ such that~$\bbP[|\eis(w)| / |w| \to \gamma_1] = 1$.
\end{corollary}

Moreover, since~$|\is^{k+1}(w)| \leqslant |\is^k(w)|$
for all words~$w$ and all integers~$k \geqslant 0$,
we already know that Theorem~\ref{thm:6} holds,
with~$\gamma_k = 0$, when the terminal component of~$M$ has size one. Henceforth, we assume that this terminal
component has size at least two.

Under this assumption, let us show that
the letters of the word~$\eis(w)$ are also generated by
a Markov chain.
In order to do so,
we introduce the function~$M^+ \colon \A \to \bbR$ defined by
\[M^+(x) = \sum_{y \colon x < y} M(x,y)\]
for every letter~$x \in \A$,
and the function~$m \colon \A^+ \cdot (\varepsilon + \$) \to \bbR$
defined by
\[m(w_0 w_1 \cdots w_k) = M(w_1,w_0) M(w_2,w_1) \cdots M(w_k,w_{k-1})\]
and~$m(w \cdot \$) = m(w) \mu(w_{-1})$
for every word~$w = w_0 w_1 \cdots w_k$ in~$\A^+$.
We also define the set
\begin{align*}
\U^\wedge \eqdef \{w_0 w_1 \cdots w_\ell \in \A^+ \cdot (\varepsilon + \$) \colon & M^+(w_0) > 0 \text{ and } \\
& \exists k \leqslant \ell, w_0 \leqslant
\ldots \leqslant w_{k-1} < w_k \geqslant 
\ldots \geqslant w_{\ell-1} > w_\ell\}.
\end{align*}

\begin{lemma}
\label{lem:9}
\label{lem:appendix:lemmaA}
The letters of the word~$\eis(w)$
are generated from right to left
by the Markov chain~$(\ovl{M},\ovl{\mu})$
with set of states~$\U^\wedge$,
whose initial distribution is defined by
\[\ovl{\mu}(w) = M^+(w_0) m(w)
\mathbf{1}_{w_{-1} = \$}\]
for every word~$w \in \U^\wedge$,
and whose transition matrix is defined by
\[
\ovl{M}(w,w') = \dfrac{M^+(w'_0)}{M^+(w_0)} \mathbf{1}_{w_0 = w'_{-1}}
m(w').\]
\end{lemma}

\begin{proof}
Let~$u^{(1)}, u^{(2)}, \ldots, u^{(k)}$ be unimodal words
such that~$u^{(i)}_{-1} = u^{(i+1)}_0$ for all~$i \leqslant k-1$.
These are the~$k$ rightmost letters of the word~$\eis(w)$ if and only if there exists
a letter~$x \in \A$ such that~$x > u^{(1)}_0$
and~$w \cdot \$$ ends with the suffix~$x \cdot u^{(1)} \cdot u^{(2)}_{1 \cdots} \cdot u^{(3)}_{1 \cdots} \cdots u^{(k)}_{1 \cdots}$,
which happens with probability
\[
\mathbf{P}_x \eqdef 
M\big(u^{(1)}_0,x\big) m\big(u^{(1)}\big) 
m\big(u^{(2)}\big) \cdots m\big(u^{(k-1)}\big) m\big(u^{(k)}\big) \mathbf{1}_{u^{(k)}_{-1} = \$}.\]
Summing these probabilities~$\mathbf{P}_x$ for all~$x > u^1_0$, we observe that~$u^{(1)}, u^{(2)}, \ldots, u^{(k)}$
are the rightmost letters of~$\eis(w)$ with probability
\begin{align*}\mathbf{P} & =
M^+\big(u^{(1)}_0\big) m\big(u^{(1)}\big) 
m\big(u^{(2)}\big) \cdots 
m\big(u^{(k-1)}\big) m\big(u^{(k)}\big) \mathbf{1}_{u^{(k)}_{-1} = \$} \\
& =
\ovl{M}\big(u^{(2)},u^{(1)}\big) \ovl{M}\big(u^{(3)},u^{(2)}\big) \cdots 
\ovl{M}\big(u^{(k)},u^{(k-1)}\big) \ovl{\mu}\big(u^{(k)}\big).
\end{align*}

Finally, Corollary~\ref{cor:8} proves that,
if~$w$ is a left-infinite word whose letters
are generated by~$(M,\mu)$ from right to left,
the word~$\eis(w)$ is almost surely infinite.
It follows that~$\ovl{\mu}$ is indeed a probability
distribution and that~$M$ is indeed a transition matrix, i.e., that
\[
\sum_{w' \in \U^\wedge} \ovl{\mu}(w') = 1 \text{ and }
\sum_{w' \in \U^\wedge} \ovl{M}(w,w') = 1\]
for all words~$w \in \U^\wedge$.
\end{proof}

Our next move consists in proving that
the Markov chain~$(\ovl{M},\ovl{\mu})$ is \epri,
by exhibiting its stationary
distribution. To that end, we first require
the following result, which roughly states that
``almost surely, every letter of a left-infinite
word~$w$ generated by~$(M,\mu)$ belongs to a
unimodal factor of~$w$'',
and whose formal
proof can be found in Appendix~\ref{sec:a:lem:11}.

\begin{restatable}{lemma}{lemmaA}
\label{lem:11}
For all letters~$x \in \A$ such that~$M^+(x) \neq 0$, 
we have
\[
\overline{\nu}(x,\uparrow) = \sum_{w \in \U^\wedge 
\colon x = w_0} \! m(w) \overline{\nu}(w_{-1},\uparrow).\]
\end{restatable}

With this result in hand, we can now prove
Proposition~\ref{pro:9}, following the same lines of
the proofs used for Proposition~\ref{pro:7}.

\begin{proposition}
\label{pro:9}
Let~$(M,\mu)$ be an \epri Markov chain
whose terminal component
has size at least two.
The Markov chain~$(\ovl{M},\ovl{\mu})$
is \epri.
\end{proposition}

\begin{proof}
First, let~$\gamma_1$ be the constant of
Corollary~\ref{cor:8}. Theorem~\ref{thm:4.16}
proves that
\[
\gamma_1 = \sum_{(x,\uparrow) \in \overline{\X}}
\! \left(\sum_{y \in\text{\rlap{\phantom{$\overline{\X}$}}} \X \colon x < y}
M(x,y) \overline{\nu}(x,\uparrow) \right) =
\sum_{(x,\uparrow) \in \overline{\X}}
M^+(x) \overline{\nu}(x,\uparrow).\]
Then, consider the distribution~$\ovl{\nu}$ defined by
\[
\ovl{\nu}(w) = \frac{1}{\gamma_1}
M^+(w_0) m(w) 
\overline{\nu}(w_{-1},\uparrow)\]
Lemma~\ref{lem:11} proves that
\[
\sum_{w \in \U^\wedge} \ovl{\nu}(w) =
\frac{1}{\gamma_1} \sum_{x \in \A} M^+(x)
\sum_{w \in \U^\wedge 
\colon x = w_0} \! m(w) \overline{\nu}(w_{-1},\uparrow) =
\frac{1}{\gamma_1} \sum_{x \in \A} \overline{\nu}(x,\uparrow) M^+(x) = 1,\]
i.e., that~$\ovl{\nu}$ is a probability
distribution.

Moreover, for every word~$w \in \U^\wedge$,
Lemma~\ref{lem:11} also proves that
\begin{align*}
\ovl{M} \ovl{\nu}(w) & =
\frac{1}{\gamma_1}
\sum_{w' \in \U^\wedge} \mathbf{1}_{w_{-1} = w'_0} 
M^+(w_0) m(w) m(w') 
\overline{\nu}(w'_{-1},\uparrow) \\
& =
\frac{1}{\gamma_1}
M^+(w_0) m(w) \overline{\nu}(w_{-1},\uparrow) = \ovl{\nu}(w).
\end{align*}
This means that~$\ovl{\nu}$ is a stationary
probability distribution of~$(\ovl{M},\ovl{\mu})$.

This probability distribution is positive on the set~$\ovl{\X} \eqdef \U^\wedge \cap \X^\ast$
and is zero outside of that set. Since~$\ovl{\nu}$ is a
probability distribution, it follows that~$\ovl{\X} \neq \emptyset$.

Then, let~$G$ and~$\ovl{G}$ be the respective
underlying graphs of~$(M,\mu)$ and~$(\ovl{M},\ovl{\mu})$. 
We shall prove that~$\ovl{\X}$ satisfies the
requirements~(i) and~(iii) of \epri
Markov chains.

Hence, consider two words~$w$ and~$w'$ in~$\ovl{\X}$,
and let us choose letters~$x, y, z, t \in \X$ such that~$x \in (w'_{-1})^\uparrow$,~$w'_0 \in y^\downarrow$,~$z \in w_{-1}^\uparrow$ and~$w_0 \in t^\downarrow$.
The graph~$G$ contains a finite path that starts
with the letter~$x$, then the letters of~$w'$
(listed from right to left) and then the letter~$y$, and finishes with the letter~$z$,
the letters of~$w$ (listed from right to left),
and then the letter~$t$.
Writing these letters from right to left,
we obtain a word~$u$ whose leftmost unimodal factor
is~$w$ and whose second rightmost unimodal factor is~$w'$.
This proves that~$\ovl{G}$ contains a path
from~$w'$ to~$w$, i.e., that~$\ovl{\X}$ satisfies the requirement~(i).

Finally, consider some trajectory~$(\ovl{X}_n)_{n \geqslant 0}$ of the Markov chain~$(\ovl{M},\ovl{\mu})$.
Up to removing the first letter of every word (i.e., vertex)~$w \in \U^\wedge$ encountered
on this trajectory, reversing these shortened words,
and then concatenating the resulting
words, we obtain a trajectory~$(X_n)_{n \geqslant 0}$
of~$(M,\mu)$. That trajectory almost surely contains
a vertex~$x \in \X$, and will then keep visiting vertices
in~$\X$. Thus, our initial trajectory almost surely contains
a word~$\ovl{X}_n$ that is a word with a letter~$x \in \X$,
and all states~$\ovl{X}_m$ such that~$m \geqslant n+1$
will then belong to the set~$\U^\wedge \cap \X^\ast = \ovl{X}$,
thereby showing that~$\ovl{\X}$ satisfies the
requirement~(iii).
\end{proof}

\begin{proposition}
\label{pro:13}
The conclusion of Theorem~\ref{thm:6} holds,
provided that
the letters of~$w$ are generated
by an \epri Markov chain from right to left.
\end{proposition}

\begin{proof}
Let~$\ell$ be the smallest integer,
if any, such that the letters
of the word~$\eis^\ell(w)$
are not generated, from right to left,
by an \epri Markov chain whose terminal
component has size at least two.

If~$\ell \geqslant k$, or if~$\ell$ does not exist,
applying Corollary~\ref{cor:8} to the words~$w, \eis(w), \ldots, \eis^{k-1}(w)$ proves that,
for all~$i \leqslant k-1$,
there exists a positive constant~$\theta_i$ such that
\[
\bbP[|\eis^{i+1}(w)| / |\eis^i(w)| \to
\theta_i] = 1\]
when~$|\eis^i(w)| \to +\infty$.
In that case, the constant ~$\gamma_k = \theta_0 \theta_1 \cdots \theta_{k-1}$
satisfies the requirements of Theorem~\ref{thm:6}.

However, if~$\ell \leqslant k-1$, then~$\eis^\ell(w)$ is
generated by an \epri Markov chain whose
terminal component has size one, i.e., consists
in an absorbing state. In that case,
Corollary~\ref{cor:8} proves that~$|\eis^{\ell+1}(w)| / |\eis^{\ell}(w)| \to 0$ almost surely,
and thus the constant~$\gamma_k = 0$
satisfies the requirements of Theorem~\ref{thm:6}.
\end{proof}

\subsection{Generating letters from left to right}
\label{sec:left-right}

We focus now on the case where the letters
of~$w$ are generated from left to right, i.e.,
the letter~$w_k$ is the~$(k+1)$\textsuperscript{th} element
of a Markov chain~$(\bM,\bmu)$ --- we use a bold-face
version of those notations used in
Section~\ref{sec:right-left}.

The two following phenomena
make generating the letters of~$w$
from left to right harder.
First, whether an integer~$k$ is~$w$-non-decreasing depends on
the letters~$w_\ell$ for~$\ell \geqslant k$, and not
on the letters~$w_\ell$ for~$\ell \leqslant k$.
Second, we defined~$w$ as the prefix of length~$n$
of a right-infinite word~$\overline{w}$.
However, whether a given integer~$k \leqslant n-1$
is~$w$-non-decreasing may depend on~$n$ since, for
instance,~$n-1$ is \emph{never}~$w$-non-decreasing.
We overcome this second issue by
generalising the notion of non-decreasing integer
and of expanded IS-reduction to infinite words,
which allows us to use the following result.

\begin{lemma}
\label{lem:truncation}
Let~$\overline{w}$ be a right-infinite word, let~$n \geqslant 4$ be an integer, and
let~$w$ be a word such that~$n-4 \leqslant |w| \leqslant n+6$ and~$w_{0 \cdots n-5} = \overline{w}_{0 \cdots n-5}$.
Finally, let~$\lambda$ 
be the number of~$\overline{w}$-locally
minimal integers that are smaller than~$n$.
We have~$\lambda-4 \leqslant |\eis(w)|
\leqslant \lambda+6$, and~$\eis(w)_{0 \cdots \lambda-5} =
\eis(\overline{w})_{0 \cdots \lambda-5}$ if~$\lambda \geqslant 4$.
\end{lemma}

\begin{proof}
Let~$i_0 < i_1 < \ldots < i_{\lambda-1}$
the~$\overline{w}$-locally
minimal integers smaller than~$n$. By construction,
we know that~$i_j + 2 \leqslant i_{j+1}$ for all~$j \leqslant \lambda-2$. This means that~$i_{\lambda-3} \leqslant n-5$, and therefore an integer~$j < i_{\lambda-3}$ is~$\overline{w}$-locally minimal
if and only if
it is also~$w$-locally minimal.
Thus, the~$\lambda-4$ first unimodal factors
of both~$w$ and~$\overline{w}$ are the words~$\overline{w}_{i_j\ldots i_{j+1}}$, where~$0 \leqslant j \leqslant \lambda-5$.
This already
proves that~$|\eis(w)| \geqslant \lambda-4$ and that~$\eis(w)_{0 \cdots \lambda-5} =
\eis(\overline{w})_{0 \cdots \lambda-5}$.

Finally, if an integer~$j \leqslant n-5$ is locally~$w$-minimal but not locally~$\overline{w}$-minimal,
we know
that~$\overline{w}_{j-1} = w_{j-1} > w_j = 
\overline{w}_j$, and
therefore~$j$ is~$w$-non-decreasing but not~$\overline{w}$-non-decreasing. This means that~$w_{j-1} > w_j = w_{j+1} = \ldots =
w_{n-5}$, and therefore there may be at most
one such integer~$j$. 
Furthermore, since no two consecutive
integers may be~$w$-minimal, the interval~$\{n-4,n-3,\ldots,n+5\}$ contains at most five~$w$-locally minimal integers. Hence, there exist
at most six~$w$-locally minimal integers
that do not belong to the set~$\{i_0,i_1,\ldots,i_{\lambda-1}\}$.
This means that~$|\eis(w)| \leqslant \lambda+6$.
\end{proof}

Lemma~\ref{lem:truncation} allows us to approximate~$\eis(w)$ with a prefix of length~$\lambda$ of
the word~$\eis(\overline{w})$, and proves that this
approximation is of excellent quality.
Indeed, if we set~$\lambda_0 = n$, and inductively 
define~$\lambda_{i+1}$ as the number of~$\eis^i(\overline{w})$-minimal integers smaller than~$\lambda_i$, Lemma~\ref{lem:truncation} ensures that~$\lambda_i - 4 \leqslant |\eis^i(w)|
\leqslant \lambda_i + 6$.
Thus, evaluating~$|\eis^i(w)|$ amounts to
evaluating~$\lambda_i$:
this is the task on which we focus below,
which allows us to identify~$w$ with an right-infinite
word, thereby saving us from many technicalities.

The first hurdle we mentioned, which requires being
able to ``guess'' whether a given integer
will be~$w$-non-increasing, is easy to overcome
by proceeding as follows.
When generating a new letter~$a$, the corresponding
position in the word has a given probability
of being~$w$-non-decreasing, which depends only
on~$a$. Thus, we can ``guess'' whether this
position should be~$w$-non-decreasing with the
correct probability, and then stick to our guess.
Hence, once again, we transform
our Markov chain~$(\bM,\bmu)$ into another
Markov chain~$(\overline{\bM},\overline{\bmu})$
that will generate pairs of the form~$(w_i,{\updownarrow}_i)$, where~$w_i$ is the~$(i+1)$\textsuperscript{th} letter of our word~$w$,
whereas~${\updownarrow}_i = \uparrow$ if~$i$ is~$w$-non-decreasing, and~${\updownarrow}_i = \downarrow$ otherwise.
Note that, unlike its variant~$(\overline{M},\overline{\mu})$, this Markov chain never
generates pairs of the form~$(\$,\updownarrow)$,
which means that its state space is simply
a subset of~$\A \times \{\uparrow,\downarrow\}$.

Using this technique allows us to
follow the same lines of proof
as in Section~\ref{sec:right-left}.
Therefore, we will just mention some
milestone constructions and results towards
proving Theorem~\ref{thm:6},
and omit their proofs, which can be found in
Appendix~\ref{sec:a:sec:left}.

Assume here that the terminal component of the \epri Markov
chain~$(\bM,\bmu)$ has size at least two.
Before defining the new Markov chain~$(\overline{\bM},\overline{\bmu})$,
we first define functions~$\bM^\uparrow$
and~$\bM^\downarrow$ by
\[
\bM^\uparrow(x) =
\frac{1}{1 - \bM(x,x)} \sum_{y \colon x < y} \bM(x,y)
\text{ and }
\bM^\downarrow(x) =
\frac{1}{1 - \bM(x,x)} \sum_{y \colon x > y} \bM(x,y)\]
for all~$x \in \A$.
Then, the Markov chain~$(\overline{\bM},\overline{\bmu})$ uses
the set of states
\[\overline{\bS} \eqdef
\{(x,\updownarrow) \in \A \times \{\uparrow,\downarrow\}
\colon \bM^\updownarrow(x) \neq 0\},\]
the initial distribution
defined by~$\overline{\bmu}(x,\updownarrow) =
\bmu(x) \bM^{\updownarrow}(x)$
for all~$(x,\updownarrow) \in \overline{\bS}$,
and the transition matrix defined by
\[\begin{cases}
\displaystyle
\overline{\bM}\big((x,\uparrow),
(y,\Updownarrow)\big) = 
\frac{\bM^\Updownarrow(y)}{\bM^\uparrow(x)} \bM(x,y) 
& \text{if } x < y; \\
\displaystyle
\overline{\bM}\big((x,\downarrow),
(y,\Updownarrow)\big) = 
\frac{\bM^\Updownarrow(y)}{\bM^\downarrow(x)} \bM(x,y) 
& \text{if } x > y; \\
\displaystyle
\overline{\bM}\big((x,\updownarrow),
(y,\Updownarrow)\big) = 
\bM(x,x) & \text{if } x = y \text{ and }
\updownarrow=\Updownarrow; \\
\overline{\bM}\big((x,\updownarrow),(y,\Updownarrow)) = 0
& \text{otherwise}.\end{cases}\]

As expected,
when projecting every pair generated by~$(\overline{\bM},\overline{\bmu})$
onto its first
coordinate, we recover a realisation of
the Markov chain~$(\bM,\bmu)$. Furthermore,
since the word~$w$ is now assumed to be infinite,
the~$k$\textsuperscript{th} pair
generated by~$(\overline{\bM},\overline{\bmu})$
is of the form~$(a,\uparrow)$ if~$k-1$ is~$w$-non-decreasing, or~$(a,\downarrow)$ otherwise,
except if the Markov chain keeps looping around a state~$(a,\uparrow)$, which happens with probability~$0$
since the terminal component has size at least two.
In addition, this new Markov chain is,
unsurprisingly, \epri.

If the terminal component of our Markov chain
contains only one state, say~$z$, we need to adapt our
construction. For all~$x \in \A \setminus \{z\}$,
we have~$\bM(x,x) < 1$, and thus the above construction is
well-defined on such states.
Then, we just merge the two states~$(z,\uparrow)$ and~$(z,\downarrow)$ into a single
\emph{sink} state, say~$(z,\downarrow)$,
and we set~$\overline{\bM}\big((z,\downarrow),(z,\downarrow)\big) = 1$.

Fortunately, the following result does not depend on
the size of the terminal component of the Markov chain.

\begin{myprop}{\ref{pro:7}b}
\label{pro:7B}
Let~$(\bM,\bmu)$ be an \epri Markov chain.
The Markov chain~$(\overline{\bM},\overline{\bmu})$
defined above is \epri.
\end{myprop}

Hence, let us consider the function~$g \colon \overline{\bS} \times \overline{\bS} \mapsto \bbR$
defined by 
\[\begin{cases}
\displaystyle
g\big((x,\downarrow),(y,\uparrow)\big) = 1 &
\text{for all } x, y \in \A; \\
g(u,v) = 0 &
\text{in all other cases.}
\end{cases}\]
Given a realisation~$(w_0,{\updownarrow}_0),(w_1,{\updownarrow}_1),\ldots$
of the Markov chain~$(\overline{\bM},\overline{\bmu})$,
and denoting by~$w = w_0 w_1 \ldots$ the word obtained
by projecting these pairs onto their first coordinate,
an integer~$i \geqslant 1$ is~$w$-locally minimal
if and only if ~${\updownarrow}_{i-1} = \downarrow$ and~${\updownarrow}_i = \uparrow$, i.e., if~$g\big((w_{i-1},{\updownarrow}_{i-1}),
(w_i,{\updownarrow}_i)\big) = 1$.
Thus, using Theorem~\ref{thm:4.16} for
the function~$g$ and Lemma~\ref{lem:truncation}
allows us to prove a special case
of Theorem~\ref{thm:6} for~$k = 1$, which consists in
the following variant of Corollary~\ref{cor:8}.

\begin{mycor}{\ref{cor:8}b}
\label{cor:8B}
If the letters of~$w$ are generated from left to right
by an \epri Markov chain,
there exists a constant~$\gamma_1$ such that~$\bbP[\lambda_1 / \lambda_0 \to \gamma_1] = 1$
when~$\lambda_0 \to +\infty$.
\end{mycor}

We focus below on the case where the Markov chain
has a terminal component of size at least two.
In that case, we show that the letters of~$\eis(w)$ are also generated
from left to right by an \epri Markov chain.
Mimicking Section~\ref{sec:right-left}, we introduce
the function~$\bmm$ defined by
\[\bmm(w_0 w_1 \cdots w_k) =
\bM(w_0,w_1) \bM(w_1,w_2) \cdots \bM(w_{k-1},w_k)\]
for every word~$w_0 w_1 \cdots w_k$ in~$\A^\ast$.
We also define the sets \vspace{-0.5em}
\begin{align*}
\bU^\wedge \eqdef \{w_0 w_1 \cdots w_\ell \in \A^\ast \colon & \bM^\uparrow(w_\ell) > 0 \text{ and } \\
& \exists k \leqslant \ell-1, w_0 \leqslant
\ldots \leqslant w_{k-1} < w_k \geqslant w_{k+1}
\geqslant \ldots \geqslant w_{\ell-1} > w_\ell\} \\
\bV^\wedge \eqdef \{w_0 w_1 \cdots w_\ell \in \A^\ast \colon &
\exists k \leqslant \ell-1, w_0 \leqslant
\ldots \leqslant w_{k-1} \leqslant w_k \geqslant w_{k+1}
\geqslant \ldots \geqslant w_{\ell-1} > w_\ell\}.
\end{align*}

\begin{mylem}{\ref{lem:9}b}
\label{lem:9B}
The letters of the word~$\eis(w)$
are generated from left to right
by the Markov chain~$(\ovl{\bM},\ovl{\bmu})$
with set of states~$\bU^\wedge$,
whose initial distribution is defined by
\[\ovl{\bmu}(w) = \sum_{w' \in \bV^\wedge} \mathbf{1}_{w'_{-1} = w_0} \bmu(w'_0) \bmm(w') \bmm(w)
 \bM^\uparrow(w_{-1}),\]
and whose transition matrix is defined by
\[
\ovl{\bM}(w,w') = \dfrac{\bM^\uparrow(w'_{-1})}{\bM^\uparrow(w_{-1})}
\bmm(w') \mathbf{1}_{w_{-1} = w'_0}.\]
\end{mylem}

\begin{myprop}{\ref{pro:9}b}
\label{pro:9B}
Let~$(\bM,\bmu)$ be an \epri Markov chain
whose terminal component
has size at least two.
The Markov chain~$(\ovl{\bM},\ovl{\bmu})$
is \epri.
\end{myprop}

The above properties allow us to prove
the following result.

\begin{myprop}{\ref{pro:13}b}
\label{pro:13B}
The conclusion of Theorem~\ref{thm:6} holds,
provided that
the letters of~$w$ are generated
by an \epri Markov chain from left to right.
\end{myprop}

\begin{proof}
Let~$\overline{w}$ be the right-infinite word
whose letters are generated, from left to right,
by our Markov chain.
Then, let~$\ell$ be the smallest integer,
if any, such that the letters
of the word~$\eis^\ell(\overline{w})$
are \emph{not} generated, from left to right,
by an \epri Markov chain whose terminal
component has size at least two.

If~$\ell \geqslant k$, or if~$\ell$ does not exist,
applying Corollary~\ref{cor:8B} to the words~$\overline{w}, \eis(\overline{w}), \ldots,
\eis^{k-1}(\overline{w})$ proves that,
for all~$i \leqslant k-1$,
there exists a positive constant~$\theta_i$ such that~$\bbP[\lambda_{i+1} / \lambda_i \to
\theta_i] = 1$ when~$\lambda_i \to +\infty$.
In that case, the constant ~$\gamma_k = \theta_0 \theta_1 \cdots \theta_{k-1}$
satisfies the requirements of Theorem~\ref{thm:6}.

However, if~$\ell \leqslant k-1$,
then~$\eis^\ell(\overline{w})$ is
generated by an \epri Markov chain whose
terminal component has size one. In that case,~$\lambda_{\ell+1} / \lambda_\ell \to 0$ when~$\lambda_\ell \to +\infty$, and therefore
the constant~$\gamma_k = 0$
satisfies the requirements of Theorem~\ref{thm:6}.
\end{proof}

\section{Words with independent and identically
distributed letters}
\label{sec:independent-case}

Theorem~\ref{thm:6} roughly states that,
if the letters of a word~$w$ are generated
(either from left to right or from right to left)
by an \epri Markov chain~$(M,\mu)$,
and provided that~$|w|$ is large enough,
the ratio~$|\is^k(w)| / |w|$ should
be approximately equal to a given constant~$\gamma_k$
depending only on~$k$ and on the Markov chain.

If we are out of luck, the Markov chain~$(M,\mu)$ might generate one unique infinite word of the
form~$w \cdot w \cdot w \cdots$, where~$w$ is one of the
worst-case words provided in Theorem~\ref{thm:1}.
Consequently, and given an integer~$k \geqslant 0$,
it is possible to choose the Markov chain~$(M,\mu)$
in order to have the equality~$\gamma_k = 2^{-k}$. This is indeed a worst case,
given that~$\gamma_{\ell+1} \leqslant \gamma_\ell / 2$
for every Markov chain and every integer~$\ell \geqslant 0$.

A specific context that will shield us from such
bad cases, while being natural, is that of words
whose letters~$w_0, w_1,\ldots,w_{n-1}$ are 
independent and identically distributed random
variables with values in the alphabet~$\A$.
Let~$X$ be their common probability law.
We first recall a result
of~\cite{nicaud2015probabilistic}, which concerns cases
where~$\A$ is finite and~$X$ is the uniform law over~$\A$.

\begin{proposition}
\label{pro:15}
Let~$w$ be a word over a finite alphabet~$\A$, 
whose letters are sampled independently and uniformly
over~$\A$, i.e.,~$\bbP[w_i = a] = 1 / |\A|$ for all
integers~$i \leqslant |w|-1$ and all letters~$a \in \A$.
The constant~$\gamma_1$ of Theorem~\ref{thm:6}
satisfies the equality
\[\gamma_1 =
\frac{1}{3} - \frac{1}{6|\A|}.\]
\end{proposition}

This shows that, in the most simple cases,
the constant~$\gamma$ is bounded from above by~$1/3$,
although~$\gamma$ can be arbitrarily close to~$1/3$
when the cardinality of~$\A$ increases.
We prove below that this upper bound is
\emph{universal}.

\begin{proposition}
\label{pro:6}
Let~$n \geqslant 1$ be an integer, and let~$\A$ be a finite or countably infinite alphabet.
Let~$X$ be a probability law on~$\A$, let
\[\Omega \eqdef \{t \in [0,1] \colon \exists a \in \A
\text{ such that } \bbP[X < a] < t < \bbP[X \leqslant a]\}\]
be a subset of~$[0,1]$ of Lebesgue measure~$1$, and let ~$f : \Omega \mapsto \A$ be
the function such that~$f(t)$ is the letter~$a \in \A$ for which~$\bbP[X < a] < t < \bbP[X \leqslant a]$.
We extend~$f$ to a partial function~$[0,1]^n \mapsto \A^n$
by setting
\[f(u_0 u_1 \cdots u_{n-1}) = f(u_0) f(u_1) \cdots
f(u_{n-1})\] if each letter~$u_i$ belongs to~$\Omega$,
and not defining~$f$ over~$[0,1]^n \setminus \Omega^n$.

For every word~$u \in \Omega^n$, we have~$|\is(u)| \geqslant |\is(f(u))|$.
Furthermore, if the letters~$u_0,u_1,\ldots,$~$u_{n-1}$
are independent and
distributed according to the uniform law~$\bbU$
over~$[0,1]$, they almost surely belong to~$\Omega$,
and then the letters~$f(u_0),f(u_1),\ldots,f(u_{n-1})$
are also independent and distributed
according to the law~$X$.
\end{proposition}

\begin{proof}
First,~$\Omega$ is a disjoint union of countably many
intervals whose lengths~$\bbP[X = a]$ sum up to~$1$,
and thus it has Lebesgue measure~$1$.
The last sentence
of Proposition~\ref{pro:6} is then immediate.
Hence, we focus on proving that~$|\is(u)| \geqslant |\is(f(u))|$ when~$u \in \Omega^n$.

Given a word~$w$, we say that a sequence of integers~$a_1 < b_1 \leqslant a_2 < b_2
\leqslant \ldots \leqslant a_{2k} < b_{2k}$
is~$w$-\emph{alternating of size~$k$} if~$b_{2k} < |w|$,~$w_{a_i} > w_{b_i}$ for all \emph{odd} indices~$i$,
and~$w_{a_i} < w_{b_i}$
for all \emph{even} indices~$i$.
One checks easily that~$|\is(w)|$ is the largest size
of a~$w$-alternating sequence.
Since every~$f(u)$-alternating
sequence is also~$u$-alternating,
Proposition~\ref{pro:6} follows.
\end{proof}

Unfortunately, in general,
the letters of the word~$\is(u)$ are not
independent, and both inequalities
\[|\is^2(u)| < |\is^2(f(u))| \text{ and }
|\is^2(u)| > |\is^2(f(u))|\] may hold,
which prevents us from designing
simple bijection-flavoured variants of
Proposition~\ref{pro:6} for investigating the
length of~$\is^k(f(u))$.
Yet, Proposition~\ref{pro:6}
still leads to the following result.

\begin{theorem}
\label{thm:7}
For every alphabet~$\A$ and every probability law~$X$ on~$\A$, we have~$\gamma_1 \leqslant 1/3$.
\end{theorem}

\begin{proof}
Let~$u$ and~$w$ be~$n$-letter words whose letters are independent 
random variables following the laws~$\bbU$ and~$X$, as described in the statement of
Proposition~\ref{pro:6}.
Each integer~${i \in \{1,2,\ldots,n-2\}}$
is~$u$-minimal if and only if~$u_i = \min\{u_{i-1},u_i,u_{i+1}\}$, which happens
with probability~$1/3$, while~$0$ and~$n-1$ cannot
be~$u$-minimal.
It follows that
\[\bbE[|\is(w)|] \leqslant
\bbE[|\is(u)|] = (n-2)/3 \leqslant n/3\]
and, thanks to Theorem~\ref{thm:6}, that~$\gamma_1 \leqslant 1/3$.
\end{proof}

In view of Proposition~\ref{pro:15} and
Theorem~\ref{thm:7}, proving that~$\displaystyle \gamma_1 \leqslant
1/3 - 1 / (6|\A|)$ even if~$X$ is
\emph{not} uniform might be tempting.
Unfortunately,
the inequality is invalid when~$|\A| = 3$ and~$(p_1,p_2,p_3) = (3/8,1/4,3/8)$, because in that case~$\displaystyle \gamma_1 = 9/32 > 5/18 =
1/3 - 1/(6|\A|)$.

However, the case~$|\A| = 2$ is still promising.
Indeed, in that case,~$\gamma_1 = p_1 (1-p_1) \leqslant 1/4$,
and the letters of
the word~$\eis(w)$ are independent and identically
distributed, since the only constraints they are subject
to is that they should begin with the letter~$0$
and end with the suffix~$10$. Thus, we can still
use Theorem~\ref{thm:7} to evaluate the ratio~$|\is^2(w)| / |\is(w)|$, thereby
deriving the following result, which suggests
excellent performances of the IS-algorithm.

\begin{proposition}
\label{pro:18}
If~$|\A| = 2$, we have~$\gamma_1 \leqslant 1/4$ and~$\gamma_2 \leqslant 1/12$.
\end{proposition}

\section{Bounding the number of function calls}
\label{sec:number-of-calls}

In this last section, we provide a short argument
for proving that, if~$\A$ is finite and 
if the letters of the word~$w$ are generated, either from left to right or from right
to left, by a (non necessarily \epri) Markov chain~$(M,\mu)$, we should expect~$\O(\log(\log(|w|)))$ recursive function calls.
This is the object of the following result, whose formal
proof can be found in Appendix~\ref{sec:a:cor:21}.

\begin{restatable}{theorem}{corB}
\label{cor:21}
Let~$w \in \A^n$ be a word whose letters are generated
by a Markov chain. For all integers~$\ell \geqslant 0$, and provided that~$n$ is large enough, the IS-algorithm
has a probability~$\mathbf{P} \leqslant n^{-2^{\ell}}$
of performing more than
\[2 \log_2(\log_2(n)) + \ell\]
recursive function calls.
\end{restatable}

\begin{proof}[Proof idea]
Let~${(M,\mu)}$ be the Markov chain that generates the letters of~$w$.
The probability that two independent trajectories of~$M$ (whose initial distributions may differ)
coincide with each other on their~$k$ first steps
decreases exponentially fast with~$k$,
unless they get
trapped into a cycle from which they cannot escape.
However, every letter of the word~$\eis^\ell(w)$ represents
at least~$2^\ell$ letters from~$w$.
Thus, if two such letters coincide, the word~$w$ must
contain two identical subwords of length~$2^\ell$,
an event whose probability decreases
severely once~$2^\ell$ exceeds~$\log(|w|)$.

It remains to treat the case where~$w$ gets trapped into
a cycle from which it cannot escape. Again, the
probability that it would take more than~$k$ steps
to reach that cycle decreases exponentially fast with~$k$,
and, when~$\ell \geqslant \log_2(k)$,
these~$n$ steps (i.e., letters) will all be subsumed
in the same letter of the word~$\eis^\ell(w)$.
However, all the other letters of~$\eis^\ell(w)$
will coincide with each other,
and thus~$\eis^{\ell+1}(w)$ will contain at most one
letter, thereby preventing subsequent recursive calls
to the IS-algorithm.
\end{proof}

This result illustrates the fact
that detecting as soon as
possible special cases in which suffix arrays are easy to
compute (here, observing that the letters of~$w$ are pairwise
distinct) can result in dramatically decreasing the
size of the recursive call stack. However, the notion of
being a \emph{large enough} integer~$n$ heavily depends
on the Markov chain~$(M,\mu)$, as illustrated by the
worst cases studied in Section~\ref{sec:worst-case},
which can be arbitrarily well approximated by Markov chains.

\newpage

\bibliography{is-algo}

\newpage
\appendix

\section{Appendix}
\label{sec:appendix}

\subsection{Proving Lemma~\ref{lem:11}}
\label{sec:a:lem:11}

We focus here on formally proving Lemma~\ref{lem:11},
whose intuitive meaning was already given in
Section~\ref{sec:right-left}.
To that end, we first introduce new variants
of the set~$\U^\wedge$. These are the sets
\begin{align*}
\U^\medbackslash & \eqdef \{w_0 w_1 \cdots w_\ell
\in \A^\ast \cdot (\varepsilon + \$) \colon w_0 \geqslant
\ldots \geqslant w_{\ell-1} > w_\ell\} \\
\U^\medslash & \eqdef \{w_0 w_1 \cdots w_\ell
\in \A^\ast \colon w_0 \leqslant
\ldots \leqslant w_{\ell-1} < w_\ell\}
\end{align*}
of non-increasing (respectively, non-decreasing) words in~$\A^\ast \cdot (\varepsilon + \$)$
whose last two letters differ from each other.
We can now prove the following auxiliary result,
from which we will then deduce Lemma~\ref{lem:11}.

\begin{mylem}{\ref{lem:appendix:lemmaA}\hspace{0.4mm}-1}
\label{lem:10}
For all letters~$x \in \A$, we have
\[
\overline{\nu}(x,\downarrow) =
\sum_{w \in \U^\medbackslash 
\colon x = w_0}
m(w) \overline{\nu}(w_{-1},\uparrow)
\text{ and }
\overline{\nu}(x,\uparrow) =
\sum_{w \in \U^\medslash 
\colon x = w_0}
m(w) \overline{\nu}(w_{-1},\downarrow).\]
\end{mylem}

\begin{proof}
Up to reversing the order~$\leqslant$ on~$\A_\$$,
both equalities are equivalent to each other.
Hence, we focus on proving the left one.
Let~$x$ be some element of~$\X$,
let~$\hat{M}$ the \emph{reverse} transition matrix
of~$\overline{M}$, such as described
in Theorem~\ref{thm:1.22}, and let~$(Y_n)_{n \geqslant 0}$ be the Markov chain
with first element~$Y_0 = x$ and
with transition matrix~$\hat{M}$.
Then, let~$\mathbf{T}$ be the stopping time
defined as the smallest integer~$n \geqslant 1$
such that~$Y_n$ belongs to the set~$\{(y,\uparrow) \colon y \in \X\}$.
Since~$\hat{M}$ is~\epri,
the stopping time~$\mathbf{T}$ is almost surely finite.

For each word~$w \in \U^\medbackslash$ such that~$x = w_0$ and~$w_{-1}^\uparrow \neq \emptyset$,
i.e.,~$\overline{\nu}(w_{-1},\uparrow) \neq 0$,
the Markov chain~$(Y_n)_{n \geqslant 0}$
has a probability
\[
\mathbf{P}_w \eqdef 
\hat{M}\big((w_0,\downarrow),(w_1,\downarrow)\big) 
\hat{M}\big((w_1,\downarrow),(w_2,\downarrow)\big) 
\cdots 
\hat{M}\big((w_{-2},\downarrow),(w_{-1},\uparrow)\big)\]
of starting with the letters~$(w_0,\downarrow), (w_1,\downarrow), \ldots,
(w_{-2},\downarrow),(w_{-1},\uparrow)$,
in which case~$\mathbf{T} = |w|-1$.
Using Theorem~\ref{thm:1.22} and the construction
of~$\overline{M}$, we have
\[\mathbf{P}_w =
\frac{\overline{\nu}(w_{-1},\uparrow)}
{\overline{\nu}(x,\downarrow)}
M(w_1,w_0) M(w_2,w_1) \cdots M(w_{-1},w_{-2})
= \frac{m(w) \overline{\nu}(w_{-1},\uparrow)}{\overline{\nu}(x,\downarrow)}.\]
Conversely, whenever~$\mathbf{T} < +\infty$,
the Markov chain~$(Y_n)_{n \geqslant 0}$ starts with
such a sequence of letters.
Consequently, the probabilities~$\mathbf{P}_w$ sum up to~$1$,
which completes the proof.
\end{proof}

\lemmaA*

\begin{proof}
Let us associate every pair~$(u,v) \in \U^\medslash \times \U^\medbackslash$ such that~$u_{-1} = v_0$ with the word~$w \eqdef u \cdot
v_{1 \cdots} \in \U^\wedge$.
Lemma~\ref{lem:10} then proves that
\[
\overline{\nu}(x,\uparrow) = \sum_{u \in \U^\medslash 
\colon x = u_0} \! 
\left(\sum_{v \in \U^\medbackslash
\colon u_{-1} = v_0}
m(u) m(v) \overline{\nu}(v_{-1},\uparrow)\right) =
\sum_{w \in \U^\wedge \colon x = w_0}
m(w) \overline{\nu}(w_{-1},\uparrow). \qedhere\]
\end{proof}

\subsection{Proving Proposition~\ref{pro:13}b}
\label{sec:a:sec:left}

We focus here on formally proving
Proposition~\ref{pro:13}b, by providing complete proofs
of the results mentioned in
Section~\ref{sec:left-right}. This proofs had first been
omitted because of their similarity to those
of Section~\ref{sec:right-left}. Consequently, we list
below results that were mentioned explicitly in
Section~\ref{sec:left-right} (sometimes adapting
their wording) or were left implicit in
Section~\ref{sec:left-right} but whose variants
had appeared in Section~\ref{sec:right-left}.

\begin{myprop}{\ref{pro:7}b}
Let~$(\bM,\bmu)$ be an \epri Markov chain
whose terminal component
has size at least two.
The Markov chain~$(\overline{\bM},\overline{\bmu})$
defined in Section~\ref{sec:left-right}
is \epri.
\end{myprop}

\begin{proof}
Let~$\bG = (\A,\bE,\bpi)$
be the underlying graph of the Markov chain~$(\bM,\bmu)$, let~$\bX$ be its terminal component,
and let~$\bnu$ be its stationary distribution.
In addition, for all~$x \in \A$, let ~$x^\uparrow = \{y \in \X \colon x < y \text{ and }
(x,y) \in E\}$ and~$x^\downarrow = \{y \in \X \colon x > y \text{ and }
(x,y) \in E\}$.

The distribution ~$\overline{\bnu}$ on~$\overline{\bS}$ defined by~$\overline{\bnu}(x,\updownarrow) =
\bnu(x) \bM^{\updownarrow}(x)$
is a probability distribution, because
\[
\overline{\bnu}(x,\uparrow) +
\overline{\bnu}(x,\downarrow)
= \frac{1}{1-\bM(x,x)} \sum_{y \colon \! x \neq y}
\bM(x,y) \bnu(x) = \bnu(x) \tag{2}\label{eq:2}\]
for all~$x \in \A$. We also deduce from~\eqref{eq:2} that
\begin{align*}
\overline{\bM} \overline{\bnu}(x,\updownarrow) -
\bM(x,x) \overline{\bnu}(x,\updownarrow)
& = \sum_{y \colon \! x < y}
\frac{\bM^\updownarrow(x)}{\bM^\downarrow(y)} \bM(y,x) \overline{\bnu}(y,\downarrow) +
\sum_{y \colon \! x > y}
\frac{\bM^\updownarrow(x)}{\bM^\uparrow(y)} \bM(y,x) \overline{\bnu}(y,\uparrow) \\
& = \bM^\updownarrow(x) \sum_{y \colon \! x \neq y}
\bM(y,x) \bnu(y) \\
& = \bM^\updownarrow(x)
\left(\bM \bnu(x) - \bM(x,x) \bnu(x)\right)
= (1 - \bM(x,x)) \overline{\bnu}(x,\updownarrow),
\end{align*}
i.e., that~$\overline{\bM} \overline{\bnu}(x,\updownarrow) =
\overline{\bnu}(x,\updownarrow)$, for all~$(x,\updownarrow)
\in \A \times \{\uparrow,\downarrow\}$.
This means that~$\overline{\bnu}$ is a stationary distribution of~$(\overline{\bM},\overline{\bmu})$.

This probability distribution is positive
on the set~$\overline{\bX} \eqdef
\{(x,\updownarrow) \in \overline{\bS} \colon
x \in \bX\}$, and zero outside of~$\overline{\bX}$.
Since~$\overline{\bnu}$ is non-zero, it follows that~$\overline{\bX}$ is non-empty.

Now, let~$\overline{\bG}$ be the
underlying graph of~$(\overline{\bM},\overline{\bmu})$.
We shall prove that~$\overline{\bX}$ satisfies the
requirements~(i) and~(iii) of \epri
Markov chains.

Consider two states~$(x,\updownarrow)$ in~$\overline{\bX}$ and~$(z,\Updownarrow)$ in~$\overline{\bS}$.
Let~$y$ and~$t$ be letters in~$x^\updownarrow$ and~$z^\Updownarrow$, respectively.
The graph~$\bG$ contains a finite path from~$z$ to~$y$
whose second vertex is~$t$ and whose second last vertex is~$x$.
Therefore,~$\overline{\bG}$ contains a finite path
from~$(z,\Updownarrow)$ to~$(x,\updownarrow)$,
which shows that~$\overline{\bX}$ satisfies the requirement~(i).

Finally, consider a trajectory~$(\bY_n)_{n \geqslant 0}$ of~$\overline{\bM}$.
Its projection onto the first component is a trajectory
in~$\overline{\bG}$, and almost surely contains a vertex~$x \in \X$, followed by another vertex~$y$.
Thus,~$(\bY_n)_{n \geqslant 0}$ contains the vertex~$(x,\uparrow)$ if~$x < y$, or~$(x,\downarrow)$
if~$x > y$, and in both cases that vertex belongs to~$\overline{\bX}$. This shows that~$\overline{\bX}$ satisfies the requirement~(iii).
\end{proof}

\begin{mylem}{\ref{lem:9}b}
The letters of the word~$\eis(w)$
are generated from left to right
by the Markov chain~$(\ovl{\bM},\ovl{\bmu})$
with set of states~$\bU^\wedge$,
whose initial distribution is defined by
\[\ovl{\bmu}(w) = \sum_{w' \in \bV^\wedge} \mathbf{1}_{w'_{-1} = w_0} \bmu(w'_0) \bmm(w') \bmm(w)
 \bM^\uparrow(w_{-1}),\]
and whose transition matrix is defined by
\[
\ovl{\bM}(w,w') = \dfrac{\bM^\uparrow(w'_{-1})}{\bM^\uparrow(w_{-1})}
\bmm(w') \mathbf{1}_{w_{-1} = w'_0}.\]
\end{mylem}

\begin{proof}
Let~$u^{(1)}, u^{(2)}, \ldots, u^{(k)}$ be unimodal words
such that~$u^{(i)}_{-1} = u^{(i+1)}_0$
for all~$i \leqslant k-1$.
These are the~$k$ leftmost letters of the word~$\eis(w)$ if and only if there exists
a word~$v \in \bV^\wedge$,
two letters~$x,y \in \A$ and an integer~$\ell \geqslant 0$
such that~$v_{-1} = u^{(1)}_0$,~$u^{(k)}_{-1} = x < y$, and~$w$ begins with the prefix~$v \cdot u^{(1)}_{1 \cdots} \cdot u^{(2)}_{1 \cdots} \cdots u^{(k)}_{1 \cdots} \cdot x^\ell \cdot y$.
This happens with probability
\[
\mathbf{P}_{v,x^{\ell-1} \cdot y} \eqdef 
\bmu(v_0) \bmm(v) \bmm(u^{(1)}) \bmm(u^{(2)}) \cdots
\bmm(u^{(k)}) \bM(x,x)^\ell \bM(x,y).\]
Summing these probabilities for all~$v$,~$y$ and~$\ell$, we observe that~$u^{(1)}, u^{(2)}, \ldots, u^{(k)}$
are the left letters of~$\eis(w)$ with probability
\begin{align*}
\mathbf{P} & = \sum_{v \in \bV^\wedge} \mathbf{1}_{v_{-1} = w_0}
 \bmu(v_0) \bmm(v) \bmm\big(u^{(1)}\big) \bmm\big(u^{(2)}\big) \cdots
\bmm\big(u^{(k)}\big) \bM^\uparrow\big(u^{(k)}_{-1}\big) \\
& = \ovl{\bmu}\big(u^{(1)}\big)\ovl{\bM}\big(u^{(1)},u^{(2)}\big)
\ovl{\bM}\big(u^{(2)},u^{(3)}\big) \cdots 
\ovl{\bM}\big(u^{(k-1)},u^{(k)}\big).
\end{align*}

Finally, Corollary~\ref{cor:8B} proves that,
if~$w$ is a right-infinite word whose letters
are generated by~$(\bM,\bmu)$ from left to right,
the word~$\eis(w)$ is almost surely infinite.
It follows that~$\ovl{\bmu}$ is indeed a probability
distribution that~$\ovl{\bM}$ is indeed a transition matrix, i.e.,
that
\[\sum_{w' \in \bU^\wedge} \ovl{\bmu}(w') = 1 \text{ and }
\sum_{w' \in \bU^\wedge} \ovl{\bM}(w,w') = 1\]
for all words~$w \in \bU^\wedge$.
\end{proof}

Then, we adapt Lemma~\ref{lem:10},
which requires introducing variants of the
sets~$\U^\medbackslash$ and~$\U^\medslash$
of Section~\ref{sec:right-left}.
These variants are the sets
\begin{align*}
\bU^\medbackslash & \eqdef \{w_0 w_1 \cdots w_\ell
\in \A^\ast \colon w_0 < w_1 \geqslant w_2 \geqslant
\ldots \geqslant w_{\ell-1}\} \\
\bU^\medslash & \eqdef \{w_0 w_1 \cdots w_\ell
\in \A^\ast \colon w_0 > w_1 \leqslant w_2 \leqslant
\ldots \leqslant w_{\ell-1}\}.
\end{align*}

\begin{mylem}{\ref{lem:10}b}
\label{lem:10B}
For all letters~$x \in \A$, we have
\[
\bnu(x) =
\sum_{w \in \bU^\medbackslash 
\colon x = w_{-1}}
\bnu(w_0) \bmm(w) \text{ and }
\bnu(x) =
\sum_{w \in \bU^\medslash 
\colon x = w_{-1}}
\bnu(w_0) \bmm(w).\]
\end{mylem}

\begin{proof}
Up to reversing the order~$\leqslant$ on~$\A$,
both equalities are equivalent to each other.
Hence, we focus on proving the left one.
Let~$x$ be some element of~$\bX$, let~$\hat{\bM}$ be the \emph{reverse} transition matrix
of~$\bM$, such as described
in Theorem~\ref{thm:1.22}, and let~$(\bY_n)_{n \geqslant 0}$ be the Markov chain
with first element~$\bY_0 = x$ and
with transition matrix~$\hat{\bM}$.
Finally, let~$\mathbf{T}$ be the stopping time
defined as the smallest integer~$n \geqslant 1$
such that~$\bY_n < \bY_{n-1}$.
Since~$\hat{\bM}$ is \epri, ~$\mathbf{T}$ is almost surely finite.

For each word~$w \in \bU^\medbackslash$ such that~$x = w_{-1}$,
the Markov chain~$(\bY_n)_{n \geqslant 0}$
has a probability
\[
\mathbf{P}_w \eqdef 
\hat{\bM}(w_{-1},w_{-2}) \cdots
\hat{\bM}(w_2,w_1) \hat{\bM}(w_1,w_0)\]
of starting with the letters~$w_{-1}, \ldots, w_2, w_1, w_0$,
in which case~$\mathbf{T} = |w|-1$.
Theorem~\ref{thm:1.22} thus proves that
\[
\mathbf{P}_w =
\frac{\bnu(w_0)}
{\bnu(w_{-1})}
\bM(w_0,w_1) \bM(w_1,w_2) \cdots \bM(w_{-2},w_{-1})
= \frac{\bmm(w) \bnu(w_0)}{\bnu(x)}.\]
Conversely, whenever~$\mathbf{T} < 0$,
the Markov chain~$(\bY_n)_{n \geqslant 0}$ starts with
such a sequence of letters.
Consequently, the probabilities~$\mathbf{P}_w$ sum up to~$1$,
which completes the proof.
\end{proof}

Let us now introduce the function~$\bnu^+ \colon \A \to \bbR$ defined by
\[
\bnu^+(x) = \sum_{y \colon x < y} \bnu(y) \bM(y,x)\]
for every letter~$x \in \A$.

\begin{mylem}{\ref{lem:11}b}
\label{lem:11B}
For all letters~$x \in \A$, 
we have
\[
\bnu^+(x) = \sum_{w \in \bU^\wedge 
\colon x = w_{-1}} \! \bnu^+(w_0) \bmm(w).\]
\end{mylem}

\begin{proof}
We associate every pair~$(u,v) \in \U^\medslash \times \U^\medbackslash$ such that~$u_{-1} = v_0$ and~$v_{-1} > x$
with the pair~$(y,w) \eqdef (u_0,u_{1 \cdots} \cdot v_{1 \cdots} \cdot x) \in \A \times \U^\wedge$, which is such that~$y > w_0$. This association is bijective, and thus
Lemma~\ref{lem:10B} proves that
\begin{align*}
\bnu^+(x) & = \sum_{y \colon x < y} \bnu(y) \bM(y,x) =
\sum_{y \colon x < y} \left(\sum_{v \in \bU^\medbackslash \colon v_{-1} = y} \left(\sum_{u \in \bU^\medslash \colon u_{-1} = v_0}  \bnu(u_0) \bmm(u) \bmm(v) \right)\right) \\
& =
\sum_{w \in \bU^\wedge \colon x = w_{-1}}
\bnu^+(w_0) \bmm(w). \qedhere
\end{align*}
\end{proof}
\bigskip

\begin{myprop}{\ref{pro:9}b}
Let~$(\bM,\bmu)$ be an \epri Markov chain
whose terminal component
has size at least two.
The Markov chain~$(\ovl{\bM},\ovl{\bmu})$
is \epri.
\end{myprop}

\begin{proof}
First, let~$\gamma_1$ be the constant of
Corollary~\ref{cor:8B}. Theorem~\ref{thm:4.16}
proves that
\[
\gamma_1 = \sum_{x \in \A} \bnu^{+}(x) \bM^\uparrow(x).\]
Then, consider the distribution~$\ovl{\bnu}$ defined by
\[
\ovl{\bnu}(w) = \frac{1}{\gamma_1}
\bnu^+(w_0) \bmm(w) \bM^\uparrow(w_{-1}).\]
Lemma~\ref{lem:11B} proves that
\[
\sum_{w \in \bU^\wedge} \ovl{\bnu}(w) =
\sum_{y \in \A} \left(\sum_{w \in \bU^\wedge \colon y=w_{-1}}
\ovl{\bnu}(w) \right) =
\frac{1}{\gamma_1} \sum_{y \in \A} \bnu^+(y) \bM^\uparrow(y) = 1,\]
i.e., that~$\ovl{\bnu}$ is a probability
distribution.

Moreover, for every word~$w \in \bU^\wedge$,
Lemma~\ref{lem:11B} proves that
\[
\ovl{\bM} \ovl{\bnu}(w) =
\frac{1}{\gamma_1}
\sum_{w' \in \U^\wedge \colon \! w'_{-1} = w_0}
\bnu^+(w'_0) \bmm(w' \cdot w) \bM^\uparrow(w_{-1})
=
\frac{1}{\gamma_1} \bnu^+(w_0) \bmm(w) \bM^\uparrow(w_{-1})
= \ovl{\bnu}(w).\]
This means that~$\ovl{\bnu}$ is a stationary
probability distribution of~$(\ovl{\bM},\ovl{\bmu})$.

This probability distribution is positive on the set
\[\ovl{\bX} \eqdef
\{w \in \bU^\wedge \cap \bX^\ast
\colon \exists x \in \bX, x > w_0 \text{ and } \bmm(x \cdot w) \neq 0\}\]
and zero outside of that set. Since~$\ovl{\bnu}$ is a
probability distribution, it follows that~$\ovl{\bX} \neq \emptyset$.

Then, let~$\bG$ and~$\ovl{\bG}$ be the respective
underlying graphs of~$(\bM,\bmu)$ and~$(\ovl{\bM},\ovl{\bmu})$. 
We shall prove that~$\ovl{\bX}$ satisfies the
requirements~(i) and~(iii) of \epri
Markov chains.

Hence, consider two words~$w$ and~$w'$ in~$\ovl{\bX}$,
and let us choose letters~$x, y, z, t \in \X$ such that~$x \in (w'_{-1})^\uparrow$,~$w'_0 \in y^\downarrow$,~$z \in w_{-1}^\uparrow$ and~$w_0 \in t^\downarrow$.
The graph~$\bG$ contains a finite path that starts
with the letter~$t$, then the letters of~$w$
(listed from left to right) and then the letter~$z$, and finishes with the letter~$y$,
the letters of~$w'$ (listed from left to right),
and then the letter~$x$.
This path forms a word~$u$ whose leftmost unimodal
factor
is~$w$ and whose second rightmost unimodal factor is~$w'$.
This proves that~$\ovl{\bG}$ contains a path
from~$w$ to~$w'$, i.e., that~$\ovl{\bX}$
satisfies the requirement~(i).

Finally, consider some trajectory~$(\ovl{\bY}_n)_{n \geqslant 0}$ of the Markov chain~$(\ovl{\bM},\ovl{\bmu})$.
Up to removing the first letter of every word (i.e., vertex)~$w \in \bU^\wedge$ encountered
on this trajectory, and then concatenating the resulting
words, we obtain a trajectory~$(\bY_n)_{n \geqslant 0}$
of~$\bM$ (for an initial distribution that may differ from~$\bmu$).
That trajectory almost surely contains
a vertex~$x \in \bX$, and will then keep visiting vertices
in~$\bX$. Thus, our initial trajectory almost surely contains
a word~$\ovl{\bY}_n$ that is a word with a letter~$x \in \bX$,
and all states~$\ovl{\bY}_m$ such that~$m \geqslant n+1$
will then belong to the set~$\bU^\wedge \cap \bX^\ast = \ovl{\bX}$,
thereby showing that~$\ovl{\bX}$ satisfies the
requirement~(iii).
\end{proof}

\subsection{Proving Theorem~\ref{cor:21}}
\label{sec:a:cor:21}

\corB*

\begin{proof}
Given a finite word~$v$
with~$v$-locally minimal integers~$i_0 < i_1 < \ldots < i_{k-1}$,
we abusively set~$i_{k+1} = |v|$ and~$v_{|v|} = \$$,
so that~$\eis(v)_\ell = v_{i_\ell \cdots i_{\ell+1}}$
for all~$\ell \leqslant k-1$.
Then, let the \emph{source} of a word~$v' = \eis(v)_{a \cdots b}$ be the word ~$v_{i_a \cdots i_{b+1}-1}$, which we also denote by~$\src(v')$, and which is a factor of~$v_{1 \cdots}$.
If two factors of~$\eis(v)$
coincide with each other, so do their sources,
and if they do not overlap with each other,
neither do their sources.
Moreover, the word~$\src(v')$ is at least twice
longer than~$v'$.

More generally, the \emph{$\ell$\textsuperscript{th}
source} of a factor~$v'$ of~$\eis^\ell(v)$,
which we denote by~$\src^\ell(v')$, is just~$v'$
itself if~$\ell = 0$, or the~$(\ell-1)$\textsuperscript{th} source of~$\src(v')$
if~$\ell \geqslant 1$.
Thus, if two letters of~$\eis^\ell(v)$ coincide
with each other, so do their ~$\ell$\textsuperscript{th} sources, which are 
non-overlapping factors of~$v_{2^\ell-1 \cdots}$
of length at least~$2^\ell$.
Moreover, since the last letter of~$\eis^\ell(v)$ is
the only one that ends with the character~$\$$,
it cannot coincide with any other letter of~$\eis^\ell(v)$.
Therefore, the~$\ell$\textsuperscript{th} sources of our two
equal letters are in fact factors of the word~$v_{2^\ell-1 \cdots |v|-2^\ell}$.

We say that~$v$ is \emph{$k$-periodic
except at borders of length~$b$}
if~${v_j = v_{j+k}}$ whenever~${b \leqslant j < j+k \leqslant |v|-b}$.
If the factor~$v_{b \cdots |v|-b}$ has exactly one letter,
none of the integers~$b+1,\ldots,|v|-b$
is locally~$v$-minimal, and thus~$|\eis(v)| \leqslant b$, thereby proving that
the word~$\eis^\ell(v)$ cannot exist whenever~$\ell \geqslant \log_2(b)+1$.
This case occurs in particular when~${k = 1}$.

Similarly, if~$|v| \leqslant 2b+3k$, the word~$\eis^\ell(v)$ cannot exist whenever~$\ell \geqslant \log_2(\max\{b,k\})+3$.

If, on the contrary, the factor~$v_{b \cdots |v|-b}$ has at least two letters
and is of length at least~$3k$,
there exists a factor~$\mathbf{f}$
of~$\eis(v)$ whose source
is a word of the form~$v_{j \ldots j+k-1}$
for some~$j$ such
that~${b \leqslant j < j+k \leqslant |v|-b}$.
Let us then write~$v$ as a concatenation of the form~$u \cdot \src(\mathbf{f})^t \cdot u'$ where~$u$ and~$u'$ have length at most~${b+k}$, and~$t$ is a
positive integer.
We can also write~$\eis(v)$ as a word of the form~$\mathbf{a} \cdot \mathbf{f}^t \cdot \mathbf{a}'$
such that~$\src(\mathbf{a})$ is a suffix of~$u$ and~$\src(\mathbf{b}) = u'$. By construction, we have
\[
\mathbf{a}| \leqslant |u| / 2 \leqslant (b+k) / 2 \text{, }
|\mathbf{f}| \leqslant |\src(\mathbf{f})| / 2 = k / 2
\text{ and }
|\mathbf{a}'| \leqslant |u'| / 2 \leqslant (b+k) / 2,\]
which means that~$\eis(v)$ is~$k'$-periodic
except at borders of length~$b'$ for some integers~$k' \leqslant k/2$ and~$b' \leqslant (b+k)/2
\leqslant \max\{b,k\}$.
Thus, an immediate induction on~$k$ proves that
the word~$\eis^\ell(v)$ cannot exist whenever~$\ell \geqslant \log_2(\max\{b,k\}) + \log_2(k) + 3$.

\medskip

Now, let~$G = (\S,E)$ be the underlying graph of
the Markov chain~$(M,\mu)$ that generates the letters of~$w$,
and let~$s = |\S|$ be the number of states of the
Markov chain.
Let~$\X$ (respectively,~$\Y$)
be the set of states~${x \in E}$ that belong
to a cyclic (respectively, non-cyclic) terminal
connected component of~$G$.
Finally, let~$\varepsilon$ be the smallest non-zero
edge weight in~$G$, i.e.,~$\varepsilon = \min\{M(x,y) \colon M(x,y) > 0\}$,
and let~$\eta = - \log_2(1-\varepsilon^s) / s > 0$.

From each state~$x \in E$, there is a path
starting at~$x$ and ending in~$\X \cup \Y$.
Furthermore, the shortest such path is of length
at most~$s$. It follows, for all~$k \geqslant 0$,
that
\[
\bbP[X_{k+s} \in \X \cup \Y \mid X_k = x] \geqslant
\varepsilon^s\]
and, more generally, that
\[
\bbP[X_m \notin \X \cup \Y] \leqslant
(1-\varepsilon^s)^{m/s-1} = 2^{-(m-s)\eta}\]
for all~$m \geqslant 0$.

Similarly, assume that~$\Y \neq \emptyset$.
Consider some state~$x \in \Y$, and let~$y \in \Y$ be a state accessible from~$x$ and with
at least two outgoing edges~$(y,z)$ and~$(y,z')$.
Then, let~$p$ be a path from~$x$ to~$y$.
The shortest such path has length at most~$s-1$.
Therefore, provided that~$X_k = x$ for some
integer~$k \geqslant 0$, the trajectory~$(X_i)_{i \geqslant k}$ has a probability at least~$\varepsilon^s$ of starting with the path~$p$ and then going to~$z$,
and a probability at least~$\varepsilon^s$ of starting
with the path~$p$ and then going to~$z'$. In particular, for each
finite sequence~$\mathbf{q}$ consisting of~$s+1$ states in~$\Y$, we have
\[
\mathbb{P}[(X_i)_{k \leqslant i \leqslant k+s} =
\mathbf{q} \mid X_k] \leqslant 1 - \varepsilon^s\]
and, more generally, if~$\mathbf{q}$ is a sequence
consisting of~$m+1$ states in~$\Y$, we have
\[
\mathbb{P}[(X_i)_{k \leqslant i \leqslant k+m} =
\mathbf{q} \mid X_k] \leqslant (1 - \varepsilon^s)^{m/s-1} = 2^{-(m-s)\eta}.\]

Finally, assume that~$w$ is a word of length~$n \geqslant 2^{16 s^2 (s + 1) + 64 s^2 / \eta}$,
and set
\[u = \log_2(n) / (4s) \text{, }
t = 2 \lfloor \log_2(u) \rfloor + \ell \text{ and }
m = 2^t - 1.\]
Since~$m \geqslant 2^{\ell-2} u^2 - 1$ and~$2^\ell u \geqslant 1$,~we~have
\[
2^{-(m-s)\eta} \leqslant
2^{-(2^{\ell-2} u^2 - s - 1)\eta} \leqslant
2^{-(2^\ell u s (s+1) + 2^{\ell+2} u s / \eta - (s + 1))\eta} \leqslant 2^{-2^{\ell+2} u s} = n^{-2^{\ell+2}}.\]

In conclusion, let us consider several
(non mutually exclusive) events:
\begin{itemize}[itemsep=0pt]
\item the event~$\mathcal{E}_1$, which occurs if~$X_m \notin \X \cup \Y$;
\item the event~$\mathcal{E}_2$, which occurs if~$X_m \in \X$;
\item for all integers~$u$ and~$v$ such that~$m \leqslant u$,~$u+m < v$ and~$v+m < n-m$,
the event~$\mathcal{F}_{u,v}$, which occurs if~$X_m \in \Y$ and~$X_{u+i} = X_{v+i}$
whenever~$0 \leqslant i \leqslant m$.
\end{itemize}

If~$\mathcal{E}_2$ happens, the word~$w$ is~$k$-periodic except at borders of length~$m$,
where~$k \leqslant s$ is the length of the
cycle of~$G$ to which~$X_m$ belongs.
Thus, in that case, the IS-algorithm cannot make
more than
\begin{align*}
\log_2(\max\{s,m\}) + \log_2(s) + 2 & =
\log_2(m) + \log_2(4s) \\
& \leqslant 2 \log_2(u) + \log_2(4s) + \ell \leqslant
2 \log_2(\log_2(n)) + \ell
\end{align*}
recursive function calls.

Then, if the IS-algorithm makes more than~$2 \log_2(\log_2(n)) + \ell \geqslant t$
recursive function calls, two letters of the word~$\eis^t(w)$ must coincide with each other.
This means that two non-overlapping
length-$m$ factors of the word~$w_{m \cdots |w|-m-1}$
must coincide with each other, and therefore that
either~$X_m \notin \Y$ or that one of the events~$\mathcal{F}_{u,v}$ must have occurred.
If~${X_m \notin \Y}$, and since~$\mathcal{E}_2$ may
not have occurred, this means that~$\mathcal{E}_1$ occurred.

Moreover, the events~$\mathcal{E}_1$ and~$\mathcal{F}_{u,v}$ are rare:
our above study proves
that~${\bbP[\mathcal{E}_1] \leqslant n^{-2^{\ell+2}}}$;
then, for all~$u$ and~$v$, the sequence~$(X_i)_{u \leqslant i \leqslant u+m}$ being fixed,
the event~$\mathcal{F}_{u,v}$ also occurs with
probability~$\bbP_{u,v} \leqslant n^{-2^{\ell+2}}$.

In conclusion, the IS-algorithm makes more than~$2 \log_2(n) + \ell$ recursive function calls with
a probability
\[\mathbf{P} \leqslant \bbP[\mathcal{E}_1] +
\sum_{u,v} \bbP[\mathcal{F}_{u,v}]
\leqslant n^2 \times n^{-2^{\ell+2}} \leqslant
n^{-2^{\ell}}. \qedhere\]
\end{proof}
\end{document}